\documentclass[journal,onecolumn]{IEEEtran}

\usepackage{multirow,bm}
\usepackage{subfig}
\usepackage{hyperref}
\usepackage{dblfloatfix}
\usepackage{tikz}
\usepackage[]{algorithm2e}
\usetikzlibrary{fit,calc,shapes,backgrounds,matrix,arrows}
\usepackage{amsmath,amsthm} % assumes amsmath package installed
\usepackage{amssymb} % assumes amssymb package installed

\newtheorem{theorem}{Theorem}
\newtheorem{proposition}[theorem]{Proposition}
\newtheorem{example}[theorem]{Example}
\newtheorem{definition}[theorem]{Definition}

\newtheorem{lemma}[theorem]{Lemma}
\newtheorem{corollary}[theorem]{Corollary}
\newtheorem{notation}[theorem]{Notation}
\newtheorem{convention}[theorem]{Convention}

\newtheorem*{theorem*}{Theorem}{\bfseries}{\itshape}
\newtheorem*{proposition*}{Proposition}{\bfseries}{\itshape}

\newcommand{\CS}[1]{\ensuremath{\mathcal{C}_{#1}}}
\newcommand{\NS}[1]{\ensuremath{\mathcal{N}_{#1}}}
\newcommand{\NN}[2]{\ensuremath{\mathcal{N}_{#1,#2}}}
\newcommand{\Na}{\ensuremath{\mathbb{N}}}
\newcommand{\Css}{\ensuremath{\bm{C}}}
\newcommand{\Ns}{\ensuremath{\bm{N}}}
\newcommand{\Cs}[1]{\ensuremath{\bm{C}_{#1}}}
\newcommand{\CN}[2]{\ensuremath{\mathcal{C}_{#1,#2}}}
\newcommand{\Hs}[2]{\ensuremath{\bm{H}_{#1,#2}}}
\newcommand{\Hss}{\bm{H}}
\newcommand{\Mat}[3]{\ensuremath{\mathcal{M}_{#1,#2}\left( #3 \right)}}
\newcommand{\word}[1]{\ensuremath{\boldsymbol{#1}}}
\newcommand{\uv}{\word{u}}
\newcommand{\vv}{\word{v}}
\newcommand{\wv}{\word{w}}
\newcommand{\supp}[1]{\ensuremath{\mathrm{Supp}(#1)}}

\newcommand{\Ps}{\mathcal{P}\left(\CS{2^m},\preceq_{SH}\right)}
\newcommand{\Po}{\mathcal{P}}
\newcommand{\Rel}{\mathrm{Rel}}

\title{Fast Reliability Ranking of Matchstick Minimal Networks}

\author{Vlad-Florin~Dr{\u a}goi, and Valeriu~Beiu, \IEEEmembership{Senior Member, IEEE}%
\thanks{V. Dr{\u a}goi is with the Faculty of Exact Sciences, ``Aurel Vlaicu'' University of Arad,
2-4 Elena Dragoi Str., 310330 Arad, Romania, and with Normandie Univ, France; UR, LITIS, F-76821 Mont-Saint-Aignan, France.}% <-this % stops a space
\thanks{V. Beiu is with the Faculty of Exact Sciences, ``Aurel Vlaicu'' University of Arad,
2-4 Elena Dragoi Str., 310330 Arad, Romania.}% <-this % stops a space
}

\begin{document}

\maketitle

\begin{abstract}
In this article, we take a closer look at the reliability of large minimal networks constructed by repeated compositions of the simplest possible networks. For a given number of devices $n=2^m$ we define the set of all the possible compositions of series and parallel networks of two devices. We then define several partial orders over this set and  study their properties. As far as we know the ranking problem has not been addressed before in this context, and this article establishes the first results in this direction. The usual approach when dealing with reliability of two-terminal networks is to determine existence or non-existence of uniformly most reliable networks. The problem of ranking two-terminal networks is thus more complex, but by restricting our study to the set of compositions we manage to determine and demonstrate the existence of a poset.

\end{abstract}

\section{Introduction}
\label{sec:introduction}
One well-known problem in information processing is that of identifying schemes that would allow maximizing reliability, while keeping resources bounded (redundancy factors). Obviously, the design-for-reliability problem becomes more challenging as the system grows larger and is required to function without interruptions for longer times. Another aspect is that enhancing/maximizing reliability should be done with a limited amount of additional (redundant) components. The number of components is the simplest and most obvious cost function, but other cost functions (also known as \textit{figures-of-merit}, or FoMs) have been proposed and used, e.g., the number of wires. It follows that \textit{design-for-reliability} is a constraint optimization problem: {\it maximize system reliability given limited resources}. The problem permeates way beyond computers into most man-made systems, while nature also relies on reliability principles/schemes at different levels (an example here being the human brain, having  $10^{11}$ neurons interconnected by $10^{15}$ synapses, working over many years).

In the following we shall restrict the scope of our discussions to computers. In this context, reliability was established by John von Neumann \cite{1952_vN}. The focus was on how one could design reliable circuits/computers using unreliable logic gates. The schemes proposed have replicated gates followed by voting and/or multiplexing, and led to regular and repetitive building blocks. Another take on this topic was put forward by Edward F. Moore and Claude E. Shannon \cite{1956_MS_1,1956_MS_2}. The major difference was that instead of starting from gates, Moore and Shannon decided to pursue their analysis starting from relays (switching devices). Their results were much more encouraging than \cite{1952_vN}. 

In this article, we analyze particular solutions for designing reliable and regular networks. One of the main motivations is that regular networks bode well with novel array-based designs, e.g., FinFETs \cite{2002_Gep}, vertical FET, gate-all-around FET, and arrays of beyond CMOS devices \cite{2016_Cou}. These can be extrapolated to wireless sensor networks, vehicular/mobile ad hoc networks, and Internet of Things \cite{R94,H_2018}. The solution we are advocating here for growing larger and more reliable networks is by \emph{combining smaller networks through compositions}. The basic building blocks we are going to use here are the smallest networks connected in series and in parallel. One of the main advantages of using series and parallel networks is that they are very easy to evaluate (as their reliability polynomials are easier to compute \cite{2003_KZ}), while compositions of series and parallel networks are inheriting this benefit.

\paragraph{Related work} Moore and Shannon were the first to propose the technique of ``composition'' for building complex networks \cite{1956_MS_1}. They proved that when a network is composed with itself $k$ times the resulting network is {\it significantly more reliable}. When $k$ tends to infinity the reliability of the repeated compositions approaches a Heaviside step function $\theta$. %Even though this seems an enabling scheme for improving reliability, we were not able to find direct applications in practical cases.

Lately, compositions of series and parallel \cite{2018_DCHGB_J}, as well as compositions of hammocks \cite{2018_DCHGB(1)} were advocated and evaluated. The results reported in \cite{2018_DCHGB_J} are promising for several reasons.

Firstly, the reliability polynomials are efficiently computable (for compositions of series and parallel). Secondly, there are series and parallel compositions which are comparable to hammocks (with respect to several different metrics for reliability), and thirdly the reliability polynomials of compositions have compact forms and are sparser than the ones for hammocks. 

 One of the open questions stated in \cite{2018_DCHGB_J} was: {\it What is the most reliable composition given a fixed number of devices?} The trivial solution, which is the worst case scenario, is to generate all compositions of a given size $n$, to compute all the associated reliability polynomials, and to compare them (by means of different FoMs). Here our aim is to give a non-trivial solution to this question by studying the relationships between compositions of series and parallel networks. For doing this we introduce several partial orders over the set of all compositions. 
 
 Recently, ordering the reliability polynomials of series and parallel compositions was investigated in  \cite{DCB18} using simulations. Here, while tackling the same type of problems, we are advancing the theoretical foundation for explaining the simulation results of \cite{DCB18}. In particular, the main contributions of this article are:

 \begin{enumerate}
 \item We give here a theoretical tool based on poset theory, in order to mathematically explain and prove the existence of an ordering relation over the set of reliability polynomials of compositions of series and parallel. 
 \item We introduce a matrix representation for a particular class of networks (introduced by Moore and Shannon in \cite{1956_MS_1}), that allows us to determine and prove several structural properties including duality of these networks. 

 \item We use the structure of the poset to answer fundamental questions regarding the reliability of compositions of series and parallel. %We details two of these applications in the next section.   
 \end{enumerate}
\paragraph{Application of the poset of compositions}

\textit{Uniformly most reliable matchstick minimal networks} The problem of finding {\it uniformly most reliable} (UMR) networks or graphs (both terms have been used interchangeably in the literature) is an interesting topic in reliability theory. This property is defined as follows: Among the set of all graphs with $\omega$ vertices and $n$ edges, a graph is UMR if its reliability polynomial is greater than the reliability polynomials of all other graphs from the same set, for every $p\in [0,1]$ (definition from \cite{2017_BGGS,1991_BLS}). Many results are known in this sense. For the {\it all-terminal reliability} problem there are certain conditions for which UMR networks exist \cite{1991_BLS,1981-K,1994_W,2000_AS}, and conditions for which UMR networks do not exist \cite{1981-K,2014_BC,1991_MCPP}. For the {\it two-terminal reliability} problem even fewer results are known \cite{2017_BGGS,BGGS18}. %\textbf{ADD latest Graves, Sun ...}

This article will focus only on a subclass of two-terminal minimal networks known as {\it matchstick minimal networks} (MMNs), without restriction on the number of vertices $\omega$. Hence, we will say that an MMN made of $n$ edges/devices is UMR if its reliability polynomial is greater than or equal to the reliability polynomials of all MMNs of $n$ edges/devices, for every $p\in [0,1]$. By studying the properties of our posets we will show that the all-parallel composition network is UMR-MMN.

\textit{Heaviside most reliable}
Reliability as per Moore and Shannon's paper \cite{1956_MS_1} should be understood not only as the connectivity ($s,t$-connectedness) but also as the non-connectivity ($s,t$-disconnectedness) of the network, as their networks were intended to replace switching devices (i.e., which have to connect and also to disconnect as needed). Hence, such networks should have their reliability polynomials close to $0$ for some interval $[0,p_0)$, and close to $1$ for the remaining interval $[p_0,1].$ This implies that these networks should have their reliability polynomials as close as possible to a shifted Heaviside step function $\theta(p-p_0)$ with $p_0$ around $1/2.$

That is why, in this paper we define the concept of {\it Heaviside most reliable} (HMR). More precisely, we will say that an MMN for which the reliability polynomial is smaller than the reliability polynomials of all other MMNs for all $p\in[0,p_0)$, and greater than the reliability polynomials of all other MMNs for $p\in [p_0,1],$ is HMR-MMN. Notice that the HMR for $p_0=0$ is UMR. With respect to this definition, we will prove that there are no HMR-MMNs for any $p_0\in (0,1)$.

\textit{Complexity of comparing MMNs}
Another possibly useful application of the results we are going to report here is that, when comparing MMNs by using the structural properties of the posets (i.e., symmetries, rank unimodality, etc.), we are able to reduce the computational complexity. More precisely, we show that results like those in \cite{2018_DCHGB_J} can be obtained by computing the reliability polynomials of a very small number of compositions. This fact brings a significant reduction of the computational complexity. 

It is well-known that in general computing the reliability polynomial of a network is \#P \cite{1979_V,1986_B}. Nevertheless, there are several subclasses of two-terminal networks for which this problem becomes tractable, such as particular ladder networks (e.g., Brech-Colbourn, fan, $K_3$, $K_4$ cylinders, etc.). Still, little is known about the complexity of this problem for MMNs. For the moment, we know that this problem can be solved in polynomial time for compositions of series and parallel \cite{2018_DCHGB_J}, and for series-parallel networks in general \cite{1982_SW}. However, the complexity of computing the reliability polynomials for hammocks is not yet known \cite{2018_CBDP}. That is why this study of posets reveals interesting relationships between MMNs, without computing their associated reliability polynomials.

This article is organized as follows. Section \ref{sec:prelim} starts with formal definitions of MMNs. We also define the main concepts, namely compositions of series and parallel, as well as their basic properties. In Section \ref{sec:prop_compo} we present structural properties of compositions, with emphasis on their duality property. Section \ref{sec:poset} is devoted to posets of reliability polynomials of compositions, while Section \ref{sec:prop_poset} describes the main characteristics of the most promising poset. In Section \ref{sec:applic_poset}~we illustrate two possible applications for UMR-MMNs and HMR-MMNs, before ending the paper with concluding remarks and future directions of research.

\section{Preliminaries}\label{sec:prelim}

\subsection{Minimal two-terminal networks}
\begin{definition}%[Two-terminal network]
Let $n$ be a strictly positive integer. We say that $\Ns$ is a two-terminal network of size $n$ if $\Ns$ is a circuit, made of $n$ identical devices, that has two distinguished contacts/terminals: an input or source $S$, and an output or terminus $T$. 
\end{definition}

\begin{figure}[!ht]
\begin{center}
\includegraphics[width=.8\textwidth]{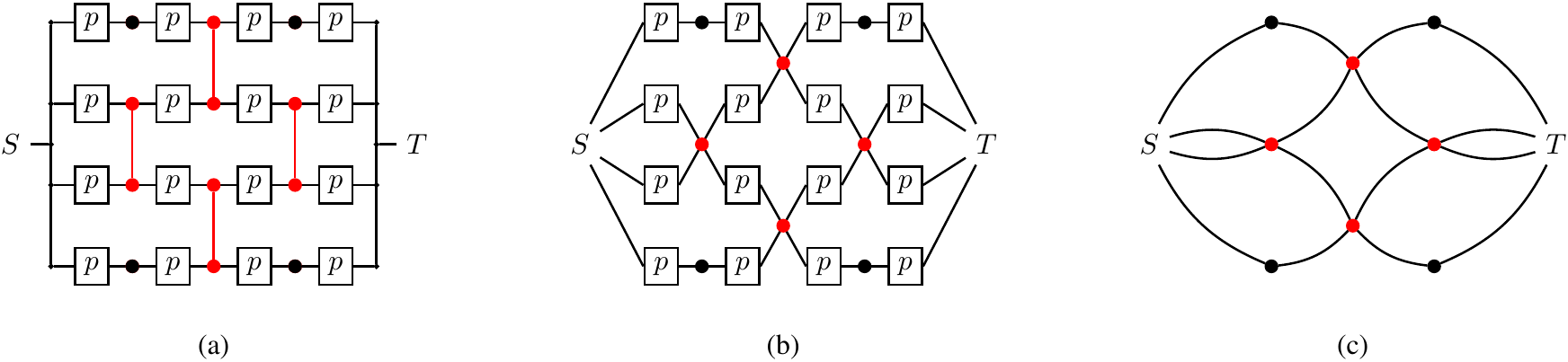}
\caption{A particular MMN, the hammock $\Hs{4}{4}$: (a) brick-wall representation \cite{1956_MS_1}; (b) hammock representation \cite{1956_MS_2}; and (c) graph representation \cite{1956_MS_1}.\label{fig:convetion_graph}}
\end{center}
\end{figure}

We emphasize that the usual way of defining and studying the reliability of a two-terminal network is by graph theoretical models \cite{1995_BCP,C87}. Even though our definition might seem different, actually one can map any circuit onto a graph by establishing a bijection between any device of a circuit and the corresponding edge of the graph, see Fig. \ref{fig:convetion_graph}. In this paper we will consider that only devices/edges can fail with independent and identical probability $q=1-p.$  %For a detailed graph theoretical definition the interested reader can consult \cite{1995_BCP,C87}. In this paper we will consider that only devices/edges can fail with independent and identical probability $q=1-p.$ 

Any two-terminal network $\Ns$ can be characterized at least by three parameters: {\it width} ($w$), {\it length} ($l$), and {\it size} $(n)$, where %  \begin{itemize}
%\item 
$w$ is the size of a ``minimal cut" separating $S$ from $T$;
%\item 
 $l$ is the size of a ``minimal path" from $S$ to $T$;
%\item 
 $n$ is related to $l$ and $w$ as $n\ge wl$ (see Theorem 3 in \cite{1956_MS_1}). If $n=wl$ we say that $\Ns$ is a minimal network. 

In the following we will restrict our investigation to a subclass of minimal two-terminal networks, which we are calling MMNs. These should not be confused with matchstick graphs \cite{H_86}, which are different structures from geometric graph theory. 
\begin{definition}\label{def:MMN}
Let $w$ and $l$ be two strictly positive integers. A two-terminal network $\Ns$ is MMN if and only if it can be designed in one of the following two ways. %\begin{itemize}
Either start with a parallel-of-series (PoS) of width $w$ and length $l$ (as in Fig. \ref{fig:hammock}~(a)) and place vertical matchsticks (wires) arbitrarily; or start with a series-of-parallel (SoP) of width $w$ and length $l$ (as in Fig. \ref{fig:hammock}~(b)) and remove vertical matchsticks (wires) arbitrarily.
%\end{itemize}
\end{definition}
 %%%%%%%%%%%%%%%%%%%%%%

A matchstick is a short wire (red) connecting two vertically adjacent nodes as in Fig. \ref{fig:convetion_graph} (a). By shrinking the matchsticks in Fig. \ref{fig:convetion_graph} (a) down to a single node (red) we obtain the ``$\times$''-crossing representation shown in Fig. \ref{fig:convetion_graph} (b).
\begin{definition} \label{def:matrix_match_min}For any MMN $\Ns$ with $w,l\ge 2$, we define its matchsticks incidence matrix $M_{\Ns}\in \Mat{w-1}{l-1}{\{0,1\}}$, as $M_{\Ns}(i,j)=1$ if there is a matchstick at position $(i,j)$ and $0$ elsewhere.
   \end{definition}
  
 For example, the four MMNs in Fig. \ref{fig:hammock} have matchstick incidence matrices

 \begin{equation*}
 M_{PoS}=\begin{pmatrix}0&0&0\\0&0&0\\0&0&0\end{pmatrix}\quad M_{SoP}=\begin{pmatrix}1&1&1\\1&1&1\\1&1&1\end{pmatrix}
 \quad M_{\Hs{4}{4}}=\begin{pmatrix}0&1&0\\1&0&1\\0&1&0\end{pmatrix}\quad M_{\Hs{4}{4}^+}=\begin{pmatrix}1&0&1\\0&1&0\\1&0&1\end{pmatrix}.
 \end{equation*}

 %%%%%%%%%%%%%%%%%%%%%%%
 \begin{convention} The convention that we adopt here is to start indexing the vectors and the matrices with $1$, respectively $(1,1).$ The elements of any set are ordered by lexicographic order, and (when the set contains integers) these are ordered with respect to the natural order on integers.
      \end{convention}

   MMNs with $w=1$ are called \textit{all-series} and do not admit a matchstick incidence matrix. This fact also holds for \textit{all-parallel} networks, that is to say MMNs with $l=1.$
The set of all MMNs of size $n=wl$ will be denoted $\NS{n},$ and we have $\NS{n}=\bigcup\limits_{w|n}^{}\NN{w}{n/w}$  (see \cite{2018_DCHGB_J}).%\end{equation}  

\begin{figure}[!hb]
\begin{center}
\includegraphics[width=\textwidth]{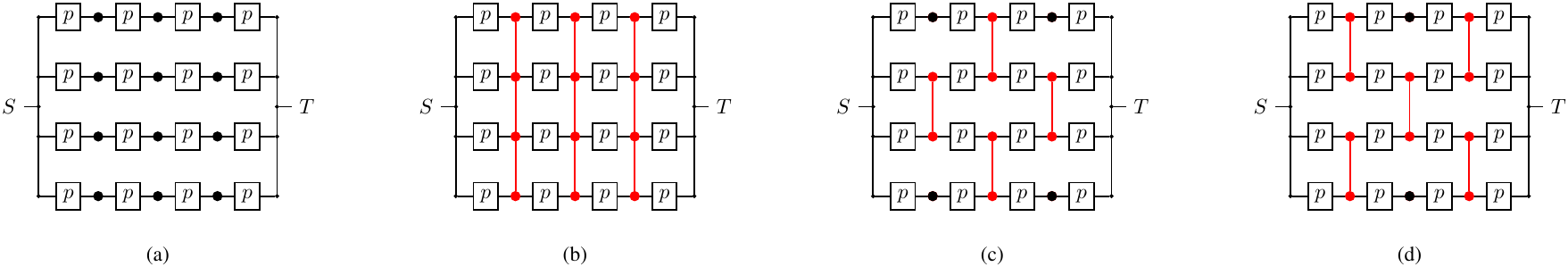}
\caption{Four representative MMNs with $w=l=4$: (a) PoS (no matchsticks); (b) SoP (all possible matchsticks); (c) $\Hs{4}{4}^{}$; and (d) $\Hs{4}{4}^{+}$.\label{fig:hammock}}
\end{center}
\end{figure}

\begin{lemma}%[\cite{1956_MS_1}]
For any strictly positive $w$ and $l$ we have 
\begin{equation}\left |\NN{w}{l}\right|=2^{(w-1)(l-1)}.\label{eq:numb_matchstick}\end{equation}
\end{lemma}

%We give here a simpler proof of this lemma using the incidence matrix of MMNs.
\begin{proof}
By Definition \ref{def:matrix_match_min} there is a one-to-one mapping between the set of MMNs of width $w$ and length $l$ and the set of binary matrices $\mathcal{M}_{w-1,l-1}(\{0,1\}).$ From this fact the combinatorial result follows.
\end{proof}

\subsection{Hammocks and compositions of series and parallel}

\paragraph{Hammock networks} MMNs with the well-known ``brick-wall'' pattern (see Fig. \ref{fig:convetion_graph}~(a)) are known as hammocks \cite{1956_MS_1,1956_MS_2,2018_CBDP}. They can be generated starting from PoS \cite{2003_KZ} (see Fig. \ref{fig:hammock}~(a)). If $w$ and $l$ are both even there are two solutions $\Hs{w}{l}$ and $\Hs{w}{l}^{+}$ (see Figs. \ref{fig:hammock}~(c) and (d)), while otherwise we are left only with $\Hs{w}{l}$ (out of all the $2^{(w-1)(l-1)}$ MMNs given by eq. \eqref{eq:numb_matchstick}).

\paragraph{Compositions of series and parallel}\label{sec:compositions}
The composition of two MMNs $\Ns_1$ and $\Ns_2$, denoted $\Ns_1\bullet\Ns_2$ is obtained by replacing each device in $\Ns_1$ by a copy of $\Ns_2$. In this article, we will consider only compositions where $\Ns_1$ and $\Ns_2$ are either two devices in series, or two devices in parallel. In order to be consistent with the existing notations from the literature we will denote two devices in series $\Css^{(0)}$, and two devices in parallel $\Css^{(1)}$. We generalize compositions of $\Css^{(0)}$ and $\Css^{(1)}$ to any $m$-length binary vector $\uv=(u_1,\dots,u_m)\in \{0,1\}^m$ as
\begin{equation}
\Css^{\uv}=\Css^{(u_1)}\bullet\dots\bullet\Css^{(u_m)}.
\end{equation}

\begin{notation} %Let $m$ be a strictly positive integer. 
We will employ similar notations as for MMNs, namely, $\Cs{2^m}$ is a network from $\CS{2^m},$ the set of all $2^m$-size compositions of $\Css^{(0)}$ and $\Css^{(1)}.$% See $\CS{2^3}$ in Fig. \ref{fig:m=3}~. 
\end{notation}

We also remember two well-known concepts, for any binary vector $\uv\in \{0,1\}^m$: 
 \begin{itemize}
\item $\supp{\uv},$ is the set of all indices corresponding to non-zero entries of $\uv$;
\item $|\uv|$ is the Hamming weight, i.e., the number of non-zero components of $\uv.$  
\end{itemize}
Notice that
\begin{equation}|\uv|=\#\supp{\uv}.
\end{equation}

For example, $\uv=(1,1,0,1)$ has $\supp{\uv}=\{1,2,4\}$ and $|\uv|=3.$

\begin{proposition}\label{pr:compo_net}
Let $w_1,w_2,l_2$ and $l_2$ be integers strictly larger than $1$, and let $\Ns_1\in \NN{w_1}{l_1}$ and $\Ns_2\in\NN{w_2}{l_2}.$ Then the composition of $\Ns_1$ and $\Ns_2$ is $\Ns=\Ns_1\bullet\Ns_2\in \NN{w_1w_2}{l_1l_2}$, with matchstick incidence matrix
\begin{equation}M_{\Ns}=
\begin{pmatrix}M_{\Ns_2}&\bm{1_{(w_2-1)\times 1}}&\dots&\dots&M_{\Ns_2}&\bm{1_{(w_2-1)\times 1}}&M_{\Ns_2}\\
\bm{0_{1\times(l_2-1)}}&M_{\Ns_1}(1,1)&\dots&\dots&\bm{0_{1\times(l_2-1)}}&M_{\Ns_1}(1,l_1-1)&\bm{0_{1\times(l_2-1)}}\\
\vdots&\vdots&\vdots&\vdots&\vdots&\vdots&\vdots\\
\bm{0_{1\times(l_2-1)}}&M_{\Ns_1}(w_1-1,1)&\dots&\dots&\bm{0_{1\times(l_2-1)}}&M_{\Ns_1}(w_1-1,l_1-1)&\bm{0_{1\times(l_2-1)}}\\
M_{\Ns_2}&\bm{1_{(w_2-1)\times 1}}&\dots&\dots&M_{\Ns_2}&\bm{1_{(w_2-1)\times 1}}&M_{\Ns_2}
\end{pmatrix}.\label{eq:matrix_compo}
\end{equation}
\begin{itemize}
\item When $w_1=1$ $M_{\Ns}$ is obtained as in \eqref{eq:matrix_compo} by tacking only the $1^{st}$ block rows containing $M_{\Ns_2}.$ %i.e., $\Ns_1$ is an all-series network,
\item When $l_1=1$ $M_{\Ns}$ is obtained as in \eqref{eq:matrix_compo} by tacking only the $1^{st}$ block column containing $M_{\Ns_2}.$   %, i.e., $\Ns_1$ is an all-parallel network,
\item When $w_2=1$ $M_{\Ns}$ is obtained as in \eqref{eq:matrix_compo} by deleting the block rows containing $M_{\Ns_2}.$ %, i.e., $\Ns_2$ is an all-series network,
\item When $l_2=1$ $M_{\Ns}$ is obtained as in \eqref{eq:matrix_compo} by deleting the block column containing $M_{\Ns_2}.$   % i.e., $\Ns_2$ is an all-parallel network, 
\end{itemize}

\end{proposition}
\begin{proposition}[\cite{2018_DCHGB_J}]\label{pr:param_comp}
Let $m$ be a strictly positive integer and $\Css^{\uv}\in \CS{2^m}.$ Then $\Css^{\uv}$ is an MMN of size $2^m,$ length $l=2^{m-|\uv|}$ and width $w=2^{|\uv|}.$ We have
\[\CS{2^m}=\bigcup_{i=0}^{m}\CN{2^i}{2^{m-i}}.\] 
\end{proposition}

 %Two particular well-known compositions are the PoS and the SoP.
\begin{proposition}Let $1\leq i \leq m$ and $\uv=(0^{m-i},1^{i})$ and $\vv=(1^i,0^{m-i}).$ Then $\Css^{\uv}$ is a SoP of $w=2^i$ and $l=2^{m-i}$ and $\Css^{\vv}$ is a SoP of $w=2^i$ and $l=2^{m-i}$. 
\end{proposition}
\section{Duality Properties of MMNs}\label{sec:prop_compo}
  
A fundamental notion mentioned in \cite{1956_MS_1} is the dual of a network, denoted as $\Ns^{\bot}.$ In order to give our main theorem for duality we introduce the bitwise complement of a binary matrix 
$M_{\Ns}\in \Mat{w-1}{l-1}{\{0,1\}}$ as \begin{equation}\overline M_{\Ns}=\bm{1}_{(w-1)\times (l-1)}\oplus M_{\Ns},
\end{equation} 
where $\bm{1}_{l\times w}$ is the all-ones matrix. 
\begin{theorem}\label{thm:dual_net}
Let $\Ns$ be a $l\times w$ MMN. If $l=1$ or $w=1$ then $\Ns$ and $\Ns^\bot$ are the all-parallel and all-series networks. % of width $w$ and its dual $\Ns^\bot$ is an all-series network of length $w.$ 
If $w,l\ge 2$ and $M_{\Ns}\in \Mat{w-1}{l-1}{\{0,1\}}$ then $\Ns^{\bot}$ is a $w\times l$ MMN 
\[M_{\Ns^{\bot}}=\left(\overline M_{\Ns}\right)^{t}.\]
%\[M_{\Ns^{\bot}}=\left(\bm{1}_{l\times w}\oplus M_{\Ns}\right)^{t},\] where $\bm{1}_{l\times w}$ is the all-ones matrix. 
\end{theorem}

In order to prove this theorem we will consider any MMN as an electrical circuit where we associate a resistance to each device, and $S$ and $T$ are connected as in Fig. \ref{fig:dual}. This is a resistor circuit which admits a dual that can be computed using Kirchoff's laws. 

\begin{figure}[!h]
\begin{center}
\includegraphics[width=0.45\textwidth]{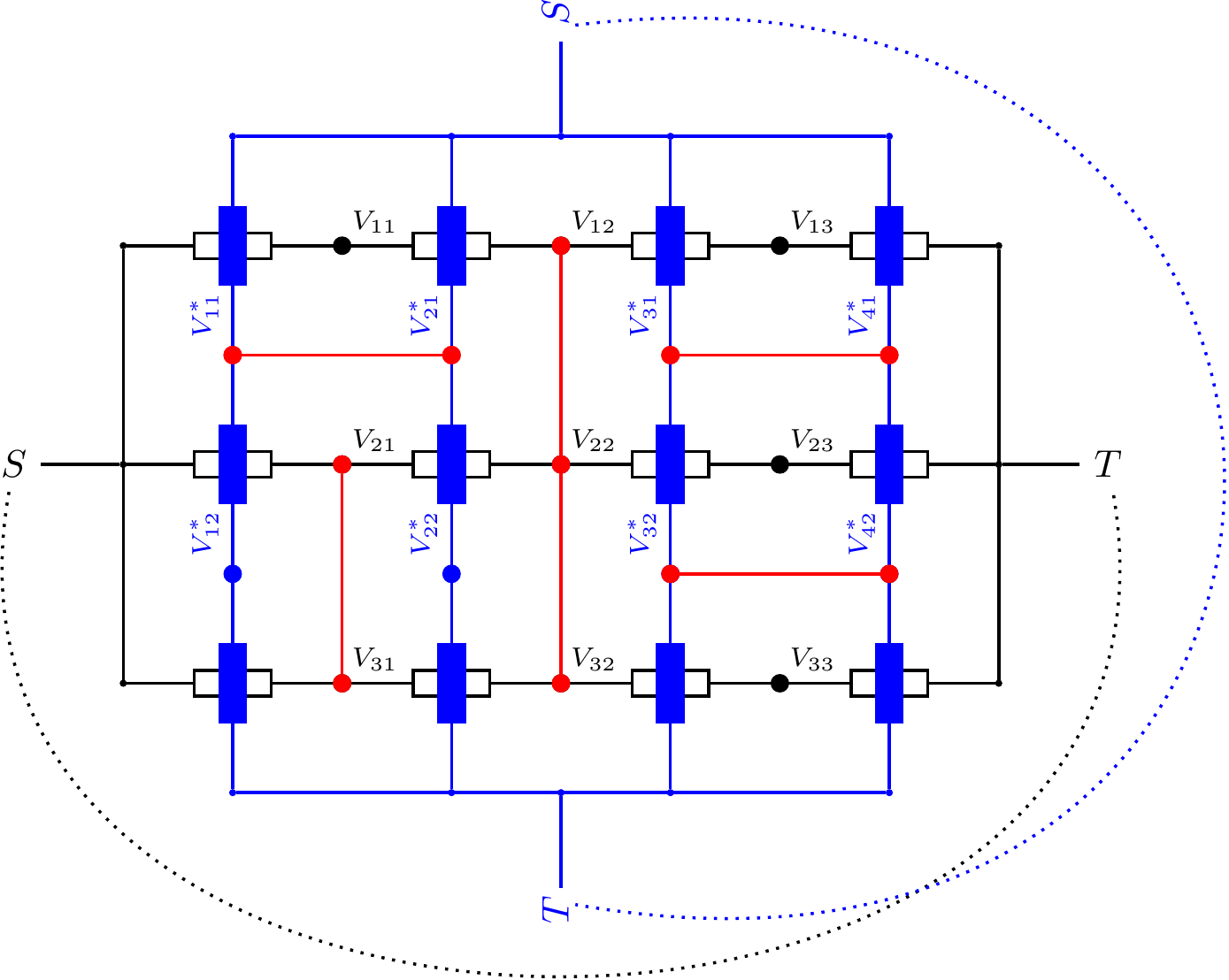}
\caption{An MMN (black) with vertical matchsticks (red), and its dual (blue) with horizontal matchsticks (also red).\label{fig:dual}}
\end{center}
\end{figure}
\begin{proof}
Let $\Ns$ be an MMN of width $w$ and length $l$. If $l=1$ or $w=1$ the result is trivial. 

Consider that $l,w\ge 2$, i.e., $\Ns$ admits an $M_{\Ns}$. A matchstick in $\Ns$ is a node in the electrical circuit and, as nodes become loops in $\Ns^{\bot}$, we will show that whenever $M_{\Ns}(i,j)=1$ we have $M_{\Ns^\bot}(j,i)=0.$ As loops become nodes in $\Ns^{\bot}$, any sequence of zeros in $M_{\Ns}$ will translate into a sequence of ones in $M_{\Ns^\bot}.$ 

Label the interior node of $\Ns$ by $V_{i,j}$ from left to right for $j$ (or from $S$ to $T$), and from top to bottom for $i$ (see Fig.~\ref{fig:dual}), where $1\leq i \leq w$ and $1\leq j\leq l-1.$ Notice that any vertically adjacent devices/resistors belong to a loop. All loops in $\Ns$ will become nodes in $\Ns^{\bot}.$ We will label these nodes $V_{i^*,j^*}^*$, where $1\leq i^*\leq l$ and $1\leq j^*\leq w-1$, $j^*$ being counted from top to bottom, and $i^*$ from left to right. Therefore, $\Ns^{\bot}$ is a network of width $l$ and length $w.$ 

For an arbitrary $(i,j)$ suppose that there is a matchstick in $\Ns$ between $V_{i,j}$ and $V_{i+1,j}$, i.e., $M_{\Ns}(i,j)=1.$ Since a matchstick corresponds to a node, this means that it will become a loop in $\Ns^\bot$ between $V_{j^*,i^*}^*$ and $V_{(j+1)^*,i^*}^*$, i.e., $M_{\Ns^\bot}(j^*,i^*)=0.$ Also, when there is no matchstick in $\Ns$ between $V_{i,j}$ and $V_{i+1,j}$, there is a matchstick in $\Ns^\bot$ between $V_{j^*,i^*}^*$ and $V_{(j+1)^*,i^*}^*.$ Since this holds for arbitrary $i$ and $j$, the proof is concluded.
\end{proof}
%From Theorem \ref{thm:dual_net} we have that $\Css^{(0)}$ is the dual of $\Css^{(1)}.$ 
We can now determine the dual of the composition of two MMNs.
\begin{lemma}\label{lem:dual_comp}
Let $\Ns_1$ and $\Ns_2$ be two MMNs. Then we have $\left(\Ns_1\bullet\Ns_2\right)^\bot=\Ns_1^\bot\bullet\Ns_2^\bot.$
\end{lemma}

This follows from Proposition \ref{pr:compo_net} and Theorem \ref{thm:dual_net}. 

\begin{proposition}
Let $m$ be a strictly positive integer and $\uv\in \{0,1\}^m.$ Then $\Css^{\overline{\uv}}$ has width $w=2^{m-|\uv|}$, length $l=2^{|\uv|}$ and 
\begin{equation}\left(\Css^{\uv}\right)^\bot=\Css^{\overline{\uv}}.
\end{equation}
\end{proposition}

\begin{proof}
Follows directly from Lemma \ref{lem:dual_comp} and Theorem \ref{thm:dual_net}.
\end{proof}

\section{Posets of Reliability}\label{sec:poset}

\subsection{Reliability polynomials}

The reliability of a two-terminal network is defined as the probability that the source $S$ and the terminus $T$ are connected (also known as $s,t$-connectness) \cite{C87}. A classical convention for the reliability polynomial is to use $\Rel(p),$ where $p \in [0,1]$ is the probability that a device is closed. Since the reliability polynomial is associated to a network $\Ns$ ($\Hss$ or $\Css$ in particular), we shall use the notation $\Rel(\Ns; p)$, which gives $\Rel(\Css; p)$ and $\Rel(\Hss; p)$ for compositions, and respectively hammocks. 

Here, we are going to rely on the following form of the reliability polynomial
\begin{equation}\Rel(\Ns;p)=\sum\limits_{i=0}^nN_i(\Ns)\;p^i(1-p)^{n-i}.\label{rel:N_form}\end{equation}

 The coefficients $N_i(\Ns)$ in eq. \eqref{rel:N_form} are integers satisfying the relation $0\le N_i(\Ns)\leq \binom{n}{i}$ for any $n$-size network $\Ns$ (see \cite{C87}), and additionally 
\begin{proposition}\label{prop:ineq_N_form}
Let $\Ns_1$ and $\Ns_2$ be arbitrary two-terminal networks of size $n$. If $N_i(\Ns_1)\leq N_i(\Ns_2),\;\forall \; 0\leq i\leq n$ then $\Rel(\Ns_1;p)\leq \Rel(\Ns_2;p).$   
\end{proposition}

Computing $\Rel(\Css; p)$ can be done using the following theorem.
\begin{theorem}[\cite{2018_DCHGB_J}]\label{thm:reliab_comp}
Let $m$ be a strictly positive integer and $\uv=(u_1,\dots,u_{m})\in \{0,1\}^m.$
Then: % reliability polynomial of the two terminal network corresponding to $\uv$ is 
\begin{equation}\Rel(\Css^{\uv};p)=\Rel(\Css^{(u_1)})\circ \dots \circ \Rel(\Css^{(u_{m})};p),\end{equation}
where $\Rel(\Css^{(0)};p)=p^2$ and $\Rel(\Css^{(1)};p)=1-(1-p)^2.$% are given by Lemma \ref{lem:r_0r_1}.  
\end{theorem}

\subsection{Partial orders}

Inspired by basic techniques from order theory (see Chapter 3 in \cite{S_2011}) we will define several partial orders for $\CS{2^m}$. We recall that a \textit{partially ordered set (poset)} is a set with a binary relation, which is \emph{reflexive, transitive} and \emph{antisymmetric}. Any pair of elements in a poset are either comparable (i.e., in relation to one another), or incomparable. In this subsection we will define several order relations for the set of compositions. For the relations that we define here it is straightforward to check reflexivity, transitivity and antisymmetry.  When comparing two MMNs we say that $\Ns_1$ is {\it more reliable} than $\Ns_2$ if

\begin{equation}\forall \; 0\le p\le 1\quad \Rel(\Ns_2;p) \le \Rel(\Ns_1;p).
\end{equation}
Using this convention we say that $\Css^{\uv}$ and $\Css^{\vv}$ are comparable, and simply write $\uv\le \vv$ or $\vv\le\uv$ if and only if for any $p\in [0,1]$ we have either $\Rel(\Css^{\uv};p)\leq \Rel(\Css^{\vv};p)$  or $\Rel(\Css^{\uv};p)\geq \Rel(\Css^{\vv};p)$, i.e., %To simplify we will use an abuse of notation and rather omit the parameter $p$ from the orders.
 \begin{equation}\uv\leq \vv\Leftrightarrow \Rel(\Css^{\uv})\leq \Rel(\Css^{\vv}).\label{eq:ord_def_1}\end{equation}

\begin{figure}[!ht]
\centering
%\resizebox{0.6\textwidth}{!}{
    \subfloat[{$m=4$}
  \label{fig:m=4}]{%

       \includegraphics[height=0.25\textwidth]{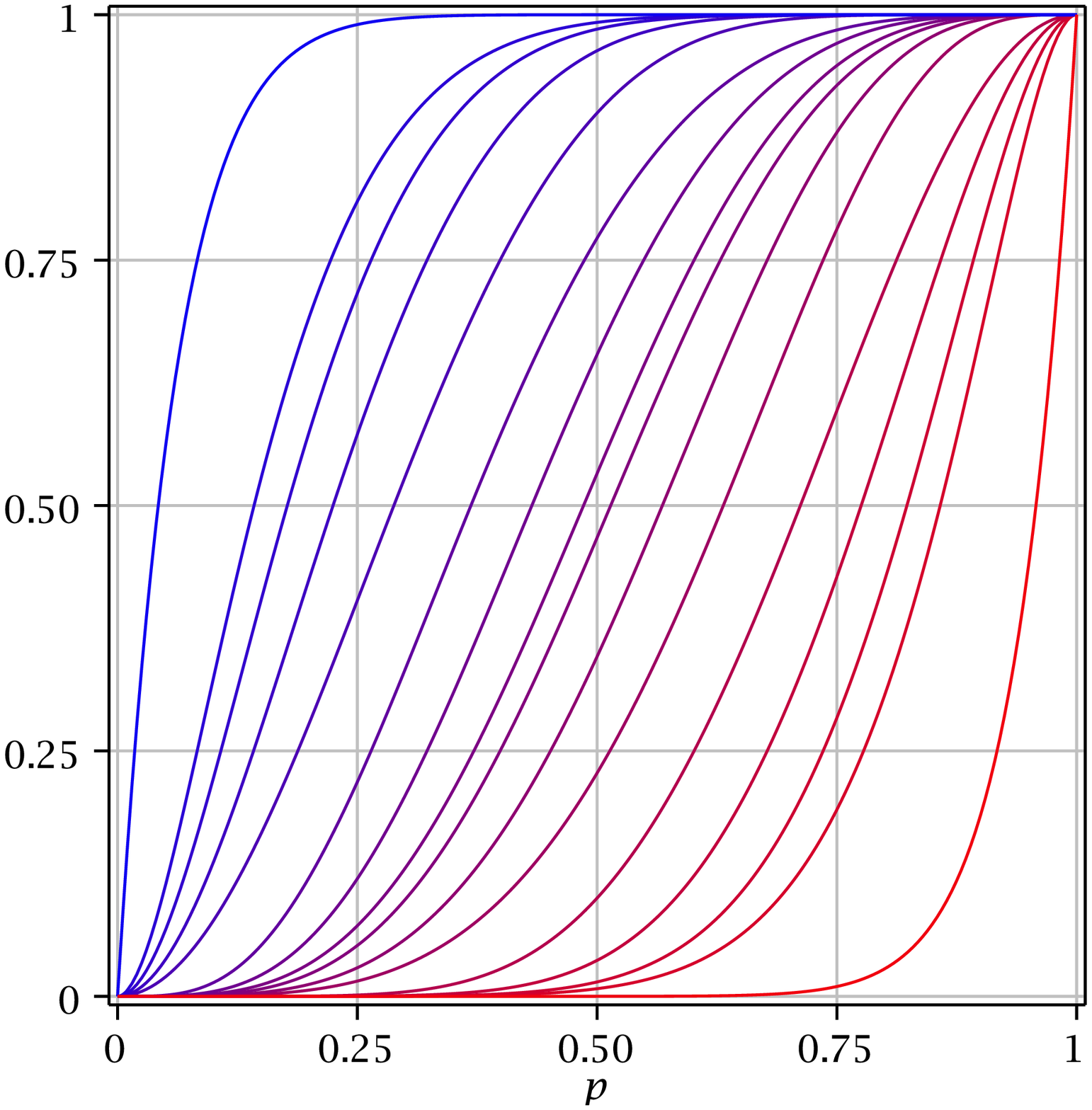}
     }
    \hfill
       \subfloat[{$m=5$}
     \label{fig:m=5}]{%
       \includegraphics[height=0.25\textwidth]{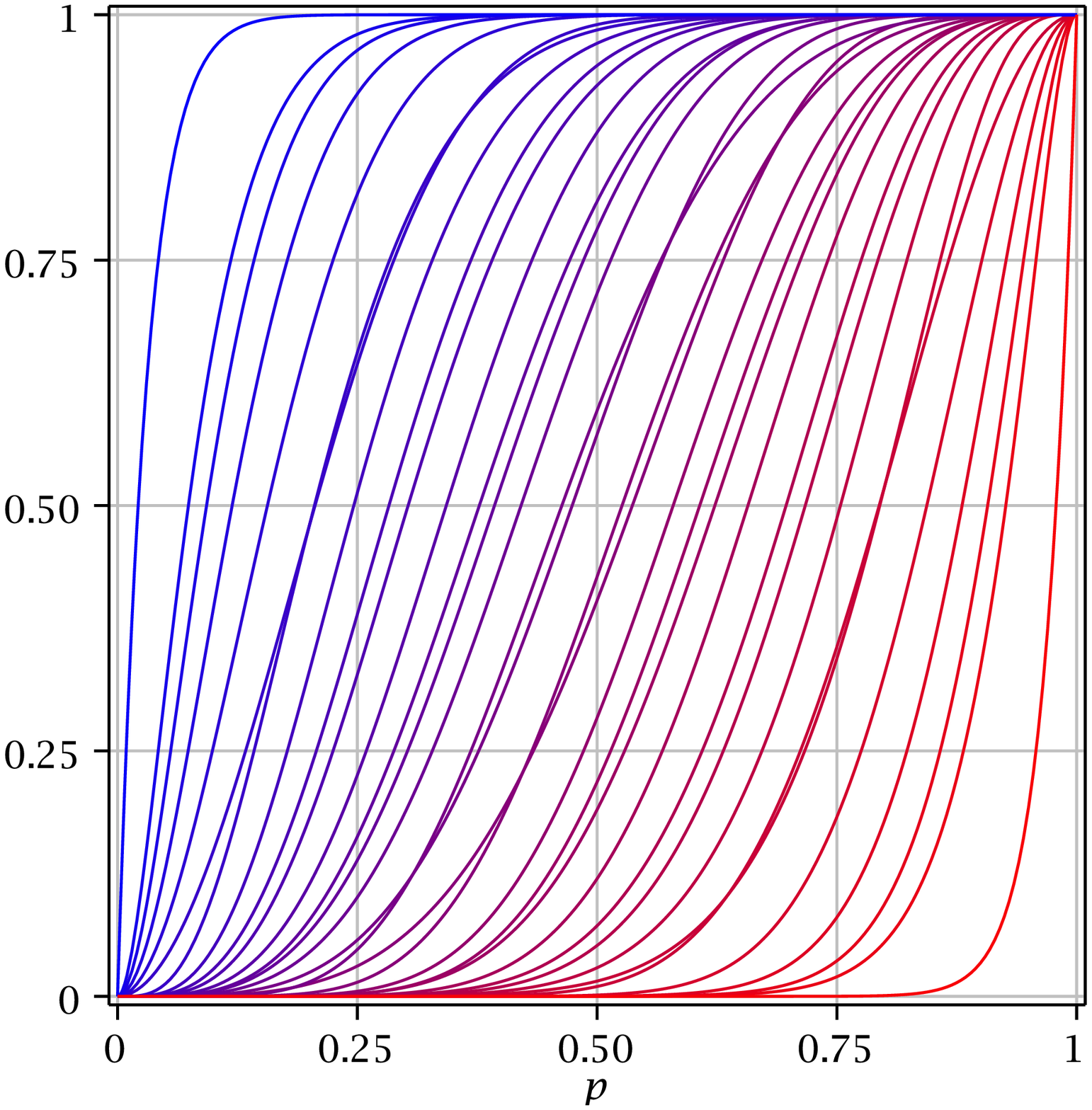}
     } 
      \hfill
     \subfloat[{$m=6$}
     \label{fig:m=6}]{%
       \includegraphics[height=0.25\textwidth]{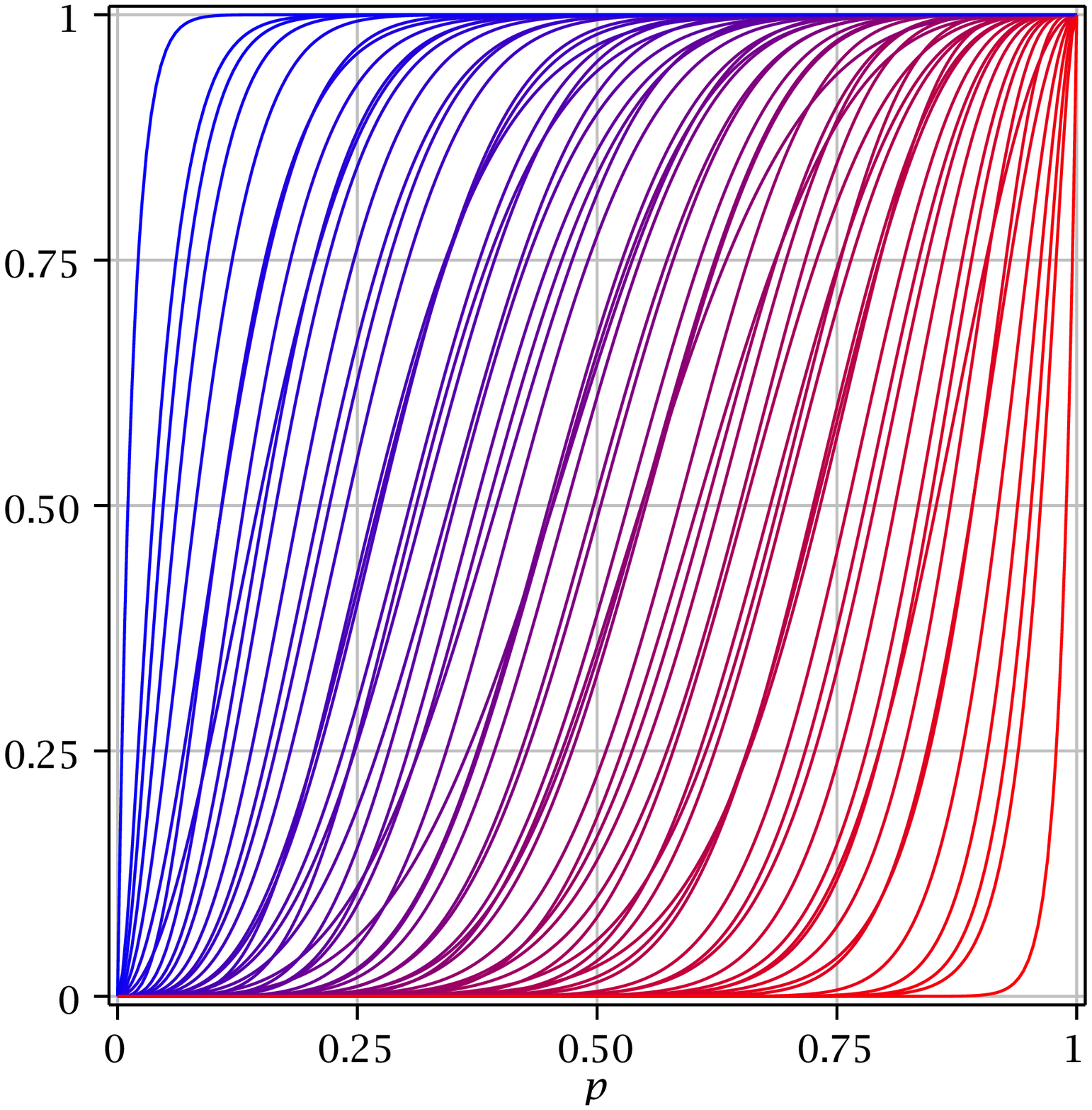}
     }
   %  }
     \caption{Reliability polynomials for all compositions in $\CS{2^m}$.}
     \label{fig:rel_comp}
   \end{figure}

\subsubsection{Preliminary simulations}

Our simulations have shown that the poset induced by the order given by eq. \eqref{eq:ord_def_1} does not look like any known poset. Moreover, there is no trivial algebraic relation over the set of binary vectors that obeys the aforementioned order. Figs. \ref{fig:m=4} and \ref{fig:m=5} plot the reliability polynomials for all the compositions when $m=4$, and respectively $m=5$. These simulations show that up to $m = 4$ the order on the compositions is total. Starting from $m = 5$ the order is partial, the compositions that are no longer comparable are $(0,0,0,0,1)$ with $(1,1,0,0,0)$, $(0,0,0,1,1)$ with $(1,1,0,1,0)$, and $(0,0,1,0,1)$ with $(1,1,1,0,0)$. As $m$ grows the number of non-comparable compositions in the poset increases (see $m=6$ in Fig. \ref{fig:m=6}).

In the next subsection, we will introduce partial orders less fine than the pointwise order, i.e., the orders to be defined here are such that if $\uv\preceq \vv$ then $\uv\leq \vv$.  
 
\subsubsection{A general order}

The first partial order that naturally comes to mind when dealing with the reliability of MMNs results from their incidence matrices. 

\begin{definition}Let $\Ns_1$ and $\Ns_2$ be two MMNs having the same $w$ and $l$. We say that $\Ns_1\preceq_M\Ns_2$ if and only if \[\forall 1\leq i\le w-1\;, \; 1\le j\le l-1\quad M_{\Ns_1}(i,j)\le M_{\Ns_2}(i,j).\]
\end{definition}
 
The meaning of $\preceq_M$ is that whenever there is a matchstick at a position $(i,j)$ in $\Ns_1$ there has to be a matchstick at the same position in $\Ns_2.$ Therefore, $\Ns_2$ has at least the same number of matchsticks as $\Ns_1$ and can also be understood as ``the matchsticks of $\Ns_1$ perfectly overlap those of $\Ns_2$.''

\begin{lemma}\label{lem:n_i_preceq_0}
Let $\Ns_1$ and $\Ns_2$ be two MMNs having the same $w$ and $l$ such that $\Ns_1\preceq_M\Ns_2.$ Then we have 
\[\forall \; 0\leq i\leq wl \; N_i(\Ns_1)\leq N_i(\Ns_2).\]
\end{lemma}

\begin{theorem}\label{thm:order_0}
Let $\Ns_1$ and $\Ns_2$ be two MMNs having the same $w$ and $l$ such that $\Ns_1\preceq_M\Ns_2.$ Then \[\Rel(\Ns_1)\leq \Rel(\Ns_2).\]
\end{theorem}

\begin{proof}
The proof follows from Lemma \ref{lem:n_i_preceq_0} and Proposition \ref{prop:ineq_N_form}.
\end{proof}

Notice that in \cite{1956_MS_1}, the authors pointed out that hammocks are ``midway'' between a PoS and a SoP with respect to  the number of matchsticks (see Fig. ~\ref{fig:hammock}). This implies the following ordering among hammocks, PoS and SoP.%That is exactly the point in the next Proposition

\begin{proposition}\label{pr:sop_ham_pos}
Let $m\ge 2$ and let $1\leq i\leq m-1$. We have 

\begin{equation*}\Rel\left(\Css^{\left(1^{i}0^{m-i}\right)}\right)\le \Rel\left(\Hs{2^{i}}{2^{m-i}}\right) \le\Rel\left(\Css^{\left(0^{m-i}1^{i}\right)}\right).
\end{equation*}
\end{proposition}
 
 \begin{proof}
 By  Proposition \ref{pr:param_comp} we check that the networks $\Css^{\left(1^{i}0^{m-i}\right)}$ (PoS) and $\Css^{\left(0^{m-i}1^{i}\right)}$ (SoP) have the same dimensions as $\Hs{2^{i}}{2^{m-i}}$, while afterwards we use $M_{\Hs{2^{i}}{2^{m-i}}}, M_{\Css^{\left(0^{m-i}1^{i}\right)}}$ and $M_{\Css^{\left(1^{i}0^{m-i}\right)}}$ in Theorem \ref{thm:order_0} to prove the result. 
 \end{proof}
 
 An even tighter inequality can be established.
 \begin{proposition}\label{pr:sop_comp_ham}
 
For $m\ge 3$ and $2\leq i\leq m-2$ we have 
\begin{equation*} \Rel\left(\Css^{\left(1^{i-1}0^{m-i-1}10\right)}\right)\le\Rel(\Hs{2^{i}}{2^{m-i}}).
\end{equation*}
\end{proposition}
  
The proof is similar to the previous one.

If some MMNs with identical $w$ and $l$ can easily be compared by means of $\preceq_M,$ what happens in the case of MMNs having different parameters? Our simulations show that several particular cases of MMNs, although having different parameters, are comparable. In particular, compositions of series and parallel are such cases. %We discuss these in the remaining of this article. 
\subsection{Partial orders for compositions}
%\subsubsection{Theoretical Results}
 
%\subsubsection{Networks of different parameters}
\begin{definition}Let $\uv$ and $\vv$ be two binary vectors of size $m.$ We define $\uv\preceq_S\vv$ if and only if $\supp{\uv}\subseteq\supp{\vv}.$
Equality holds only for $\uv=\vv.$
\end{definition}

%\subsubsection{Networks of identical parameters}
\begin{definition}Let $l\leq m$ be two strictly positive integers and $\uv,\vv$ be two binary vectors of size $m$ such that $|\uv|=|\vv|=l.$  Let $\supp{\uv}=\{s_1,\dots,s_l\}$ and $\supp{\vv}=\{t_1,\dots,t_l\},$ with $s_1<\dots<s_l$ and $t_1<\dots<t_l$. We define $\uv\preceq_H\vv$ if and only if $\forall 1\le i\le l, \;s_i\le t_i.$
\end{definition}

We combine $\preceq_S$ and $\preceq_H$ in a natural manner and define the order ``$\preceq_{SH}$'' as being equal to
\begin{itemize}\item $\preceq_S$ when comparing vectors with different Hamming weights;
\item  $\preceq_H$ when comparing vectors having the same Hamming weight.
\end{itemize}

%\subsubsection{Related work}
 The orders that we define and prove here (i.e., $\preceq_S$, $\preceq_H$ and $\preceq_{SH}$) have also been used in other fields \cite{BDOT16,M_2017}. The first one ($\preceq_S$) was proposed in the context of Boolean functions \cite{C10}, more precisely for computing the Algebraic Normal Form of a Boolean function using the Fast Mobius Transform \cite[Section 2.1]{C10}. In a completely different field, $\preceq_S$ was used  to tighten the bounds on the error block probability of a polar code designed for a binary erasure channel (\cite[Section VI]{MT09}).  In \cite{S16,BDOT16,D17}, $\preceq_{SH}$ was used to prove degradation of communication channels for polar codes, while in \cite{He_2017} and \cite{M_2017}, it was used to optimize construction of polar codes for different type of channels. In \cite{2010_GMO}, Gordon, Miller and Ostapenko used $\preceq_{SH}$ for solving the closest pair problem in large datasets by means of optimal hash functions. The order $\preceq_{SH}$ was also used in a cryptographic application \cite{BCDOT16}.  
 
 \medskip
Now, we are ready to introduce one of our main results.% which was stated without proof as Theorem 6 in \cite{DCB18}.

\begin{theorem}\label{thm:order}
Let $\uv$ and $\vv$ be two binary vectors of size $m$. Then we have
\[\uv\preceq_{SH}\vv\Rightarrow \uv\le \vv.\]%\Rel(\Css^{\uv})\le \Rel(\Css^{\vv}). \]
\end{theorem}

The proof of this theorem is given in Appendix \ref{app:A}. This was already state without proof in \cite{DCB18}.
The poset of compositions will be shorthanded as $\Ps$ and, for convenience, we will use $\uv$ instead of $\Css^{\uv}.$ 

\section{Properties of the Poset}\label{sec:prop_poset}

In order to give the structural properties of $\Ps$ we remember several fundamental concepts from poset theory. A quick bibliographic search shows that this poset is isomorphic to a well-known one, denoted as $M(n)$ in \cite[Section 4.1.2]{S_91}, where it is called {\it partitions into distinct summands}. This poset has the following main properties: {\it rank unimodal, rank symmetric} and {\it Sperner property}.

\begin{figure}[!h]
\begin{center}
\includegraphics[width=\textwidth]{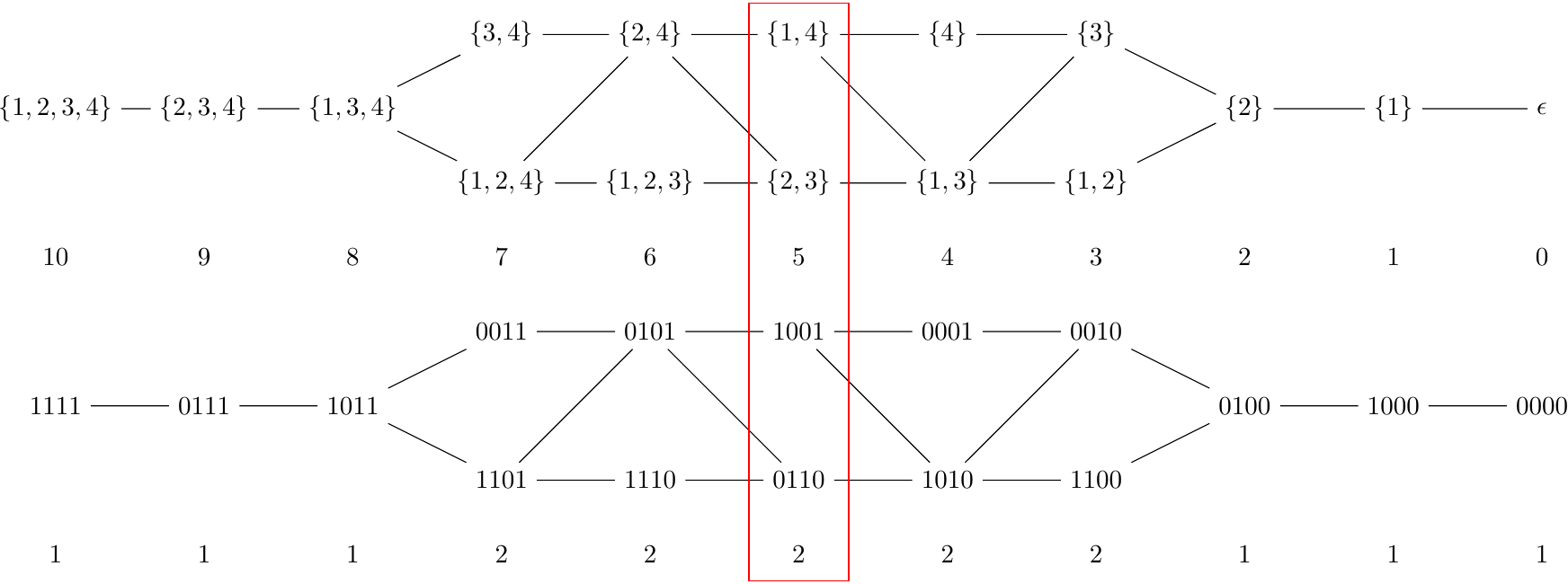}
\caption{The Hasse diagram for $\Ps$ for $m=4$, the rank values for each level of the poset, and $\#\Po_i$. \label{fig:mon_order_4}}
\end{center}
\end{figure}

\subsection{Preliminaries} 

\begin{definition}
Let $\Po$ be a poset. We call any subset of $\Po$ a chain, if and only if it is totally ordered. Any subset of $\Po$ is called an antichain if no pair of elements in it are comparable. 
\end{definition}

For example, for $m=4$ we have that $\{(0101),(1110)\}$ and $\{(1110),(1001)\}$ are two antichains, while $\\ \{(0010),(1010),(0110),(0101),(0011),(1011)\}$ is a chain (see Fig. \ref{fig:mon_order_4}).

\begin{definition}
A poset $\Po$ is bounded from above/below if there is an element $x$ of $\Po$ such that any other element of $\Po$ is smaller/larger than $x.$ When $\Po$ admits both an upper and a lower bound we simply say that $\Po$ is bounded. 
\end{definition}

\begin{definition}
Let $\Po$ be a poset. We say that $\Po$ is graded if $\Po$ can be equipped with a (rank) function $\rho:\Po\to \Na$ that satisfies: 
\begin{itemize}
\item if $x$ is minimal then $\rho(x)=0$.
\item if $y$ covers $x$ then $\rho(y)=\rho(x)+1$.
\end{itemize}
\end{definition}

Any graded poset $\Po$ can be partitioned into $\Po=\bigcup_{i=1}^r\Po_i$, where $\Po_i$ is the rank $i$ of $\Po.$ This implies that any maximal length chain in $\Po$ passes through exactly one element in each $\Po_i.$  Any $\Po_i$ is also an antichain, since all the elements of $\Po_i$ have rank $i$ and thus are not comparable. For a graded poset $\Po$ the maximum number of elements in an antichain is lower bounded by the maximum of $\#\Po_i$.   

\begin{definition}
Let $\Po$ be a graded poset with $\Po=\bigcup_{i=1}^r\Po_i.$ We say that \begin{itemize}
\item $\Po$ is rank symmetric if $\# \Po_i=\# \Po_{r-i},$ for all $i$;
\item $\Po$ is rank unimodal if the sequence $\{\#\Po_i\}_{1\le i\le r}$ is unimodal;
\item $\Po$ is Sperner if $\max_{A}\#A=\max_{i}\#\Po_i,$ where $A$ runs through the set of all antichains.  
\end{itemize}
\end{definition}

Hence, in a \emph{Sperner} poset the largest rank provides an antichain of maximum cardinality. Notice that there may exist other antichains of maximum cardinality as well. So, if $\Po$ is rank symmetric, rank unimodal, and \emph{Sperner}, then $\max_{A}\#A=\#\Po_{r/2}$ when $r$ is even, and $\max_{A}\#A=\#\Po_{\lfloor r/2\rfloor}=\#\Po_{\lceil r/2\rceil}$ when $r$ is odd.

\subsection{Applications to compositions}
\subsubsection{General properties of the poset of compositions}

\begin{proposition}
$\Ps$ is bounded, where $\uv=(0,\dots,0)$ and $\vv=(1,\dots,1)$ are the minimum, and respectively the maximum elements.
\end{proposition}

\begin{proof}
This follows directly from the definition of the order $\preceq_{SH}.$
\end{proof}

 We can also prove that $\Ps$ is graded by specifying a rank function.
\begin{proposition}
Let $m$ be a strictly positive integer. $\Ps$ is graded, where \[\forall \uv\in \{0,1\}^m\quad \rho(\uv)=\sum\limits_{i\in \supp{\uv}}^{}i.\]
\end{proposition}

\begin{proof}
First, notice that the minimum element of $\Ps$ is the all-zeros vector, which implies that $\rho(0,\dots,0)=0.$ For the second condition we need to check that if $\vv$ covers $\uv$ then $\rho(\vv)=\rho(\uv)+1.$ For any $\uv$ in the poset, $\vv$ might cover $\uv$ either by $\preceq_S$ or by $\preceq_H.$ 
\begin{itemize}
\item Suppose that $\vv$ covers $\uv$ and $\uv\preceq_H\vv$, and let $\supp{\uv}=\{s_1,\dots,s_l\}$ and $\supp{\vv}=\{t_1,\dots,t_l\}.$ This implies that $\exists \;1\le i_0\le l$ such that $s_{i_0}+1=t_{i_0}$ and $\forall \;1\le i\ne i_{0}\le l, s_i=t_i.$ Hence we have $\rho(\vv)=\rho(\uv)+1.$ 
\item Suppose that $\vv$ covers $\uv$ and $\uv\preceq_S\vv.$ By definition we have $\supp{\uv}\subset\supp{\vv}.$ Notice that unless $1\not \in \supp{\uv}$, $\vv$ can not cover $\uv$ by $\preceq_S.$ We will prove this claim by contradiction. Suppose that $1\in \supp{\uv}$ and $\exists \vv$ that covers $\uv$ by $\preceq_S.$ We distinguish two cases: 
\begin{enumerate}\item 
The first case is when $\supp{\uv}=\{1,2,\dots,l\}.$ This implies that $\{1,\dots,l\}\subset\supp{\vv}$ and the smallest $\vv$ which satisfies this condition is such that $\supp{\vv}=\{1,\dots,l+1\}.$ Now, by letting $\wv$ be such that $\supp{\wv}=\{2,3,\dots,l+1\}$, $\uv\preceq_H\wv\preceq_S\vv$, which is impossible since $\vv$ covers $\uv.$
\item The second case is when the elements in $\supp{\uv}$ are not necessarily consecutive integers from $1$ to $l.$ This implies that there are at least two elements in $\supp{\uv},$ $s_i$ and $s_{i+1}$, such that $s_i$ is the smallest element satisfying $s_{i}\le s_{i+1}-2$. The smallest $\vv$ that is $\uv\preceq_S\vv$ satisfies $\supp{\vv}=\supp{\uv}\cup\{s_{i}+1\}$. Now let $\wv$ be such that $\supp{\wv}=\{2,\dots,s_{i}+1\}\cup\{s_{i+1},\dots,s_l\}.$ Obviously, we have $\uv\preceq_H\wv\preceq_S\vv$, which is impossible.    
\end{enumerate}

Since $1\not\in \supp{\uv}$ and $1$ is the smallest element that one could add to $\supp{\uv}$ in order to obtain an element $\vv$ such that $\uv\preceq_S\vv$, the proof is concluded.
\end{itemize}
\end{proof}
\begin{theorem}[\cite{S_91}]\label{thm:pos_sperner}
Let $m$ be a strictly positive integer. Then the set of compositions ordered by $\preceq_{SH}$ is rank unimodal, rank symmetric and Sperner.
\end{theorem}

For example, when $m=4$ (see Fig. \ref{fig:mon_order_4}) the sequence of $\#\Po_i$ is $1,1,1,2,2,2,2,2,1,1,1$, which is rank symmetric and rank unimodal.

In the following we will answer four natural questions related to $\Ps$:
\begin{itemize}
\item  Which is the maximum length of a chain? 
\item How to construct a chain of maximum length?
\item Which is the middle of the poset?
\item  How to identify at least one element from the middle of the poset?
\end{itemize}

\begin{figure}[!ht]
      \centering
   %   \hfill
    \subfloat[{$m=5$}]{%
    \centering
       \includegraphics[width=0.25\textwidth,height=0.25\textwidth]{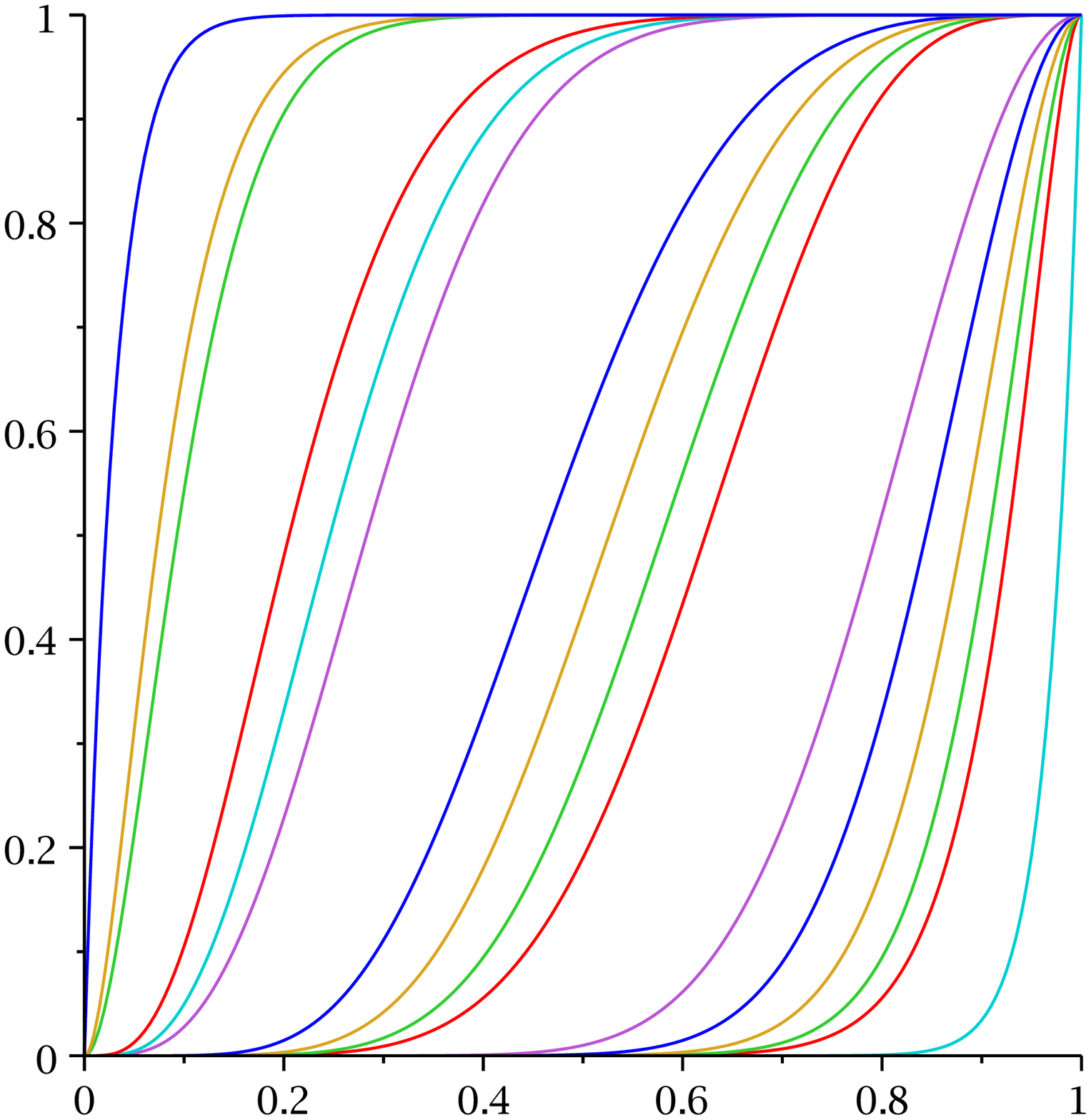}    \label{fig:chain_m=5}
     }
     \hspace{2cm}
     \subfloat[{$m=6$}]{%
     \centering
       \includegraphics[width=0.25\textwidth,height=0.25\textwidth]{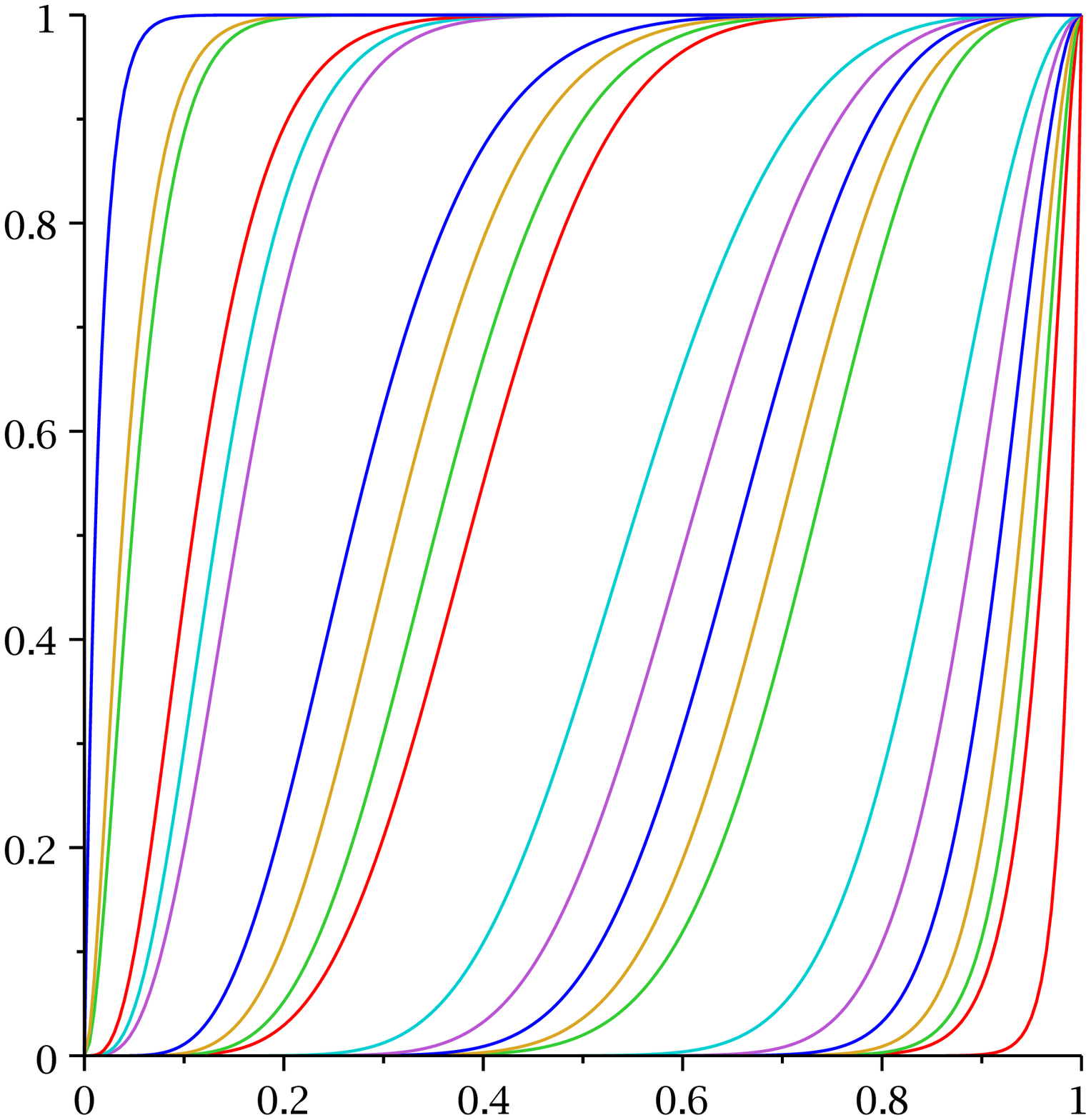}  \label{fig:chain_m=6}
     }
%    \hfill   \hfill
%     \subfloat[{$m=12$}
%     \label{fig:chain_m=12}]{%
%       \includegraphics[height=0.3\textwidth]{chain_m_12.eps}
%     }
     \caption{$\Rel(\Css^{\uv})$ for $\uv\in \mathcal{S}$, constructed using the algorithm from Proposition \ref{pr:algo_max_chain}.}
     \label{fig:chains}
   \end{figure}

  \subsubsection{Maximum length chain of the poset}
%A direct consequence of Theorem \ref{thm:pos_sperner} is 
\begin{corollary}Let $m$ be a strictly positive integer. Then the maximum length of a chain in $\Ps$ equals $\binom{m+1}{2}$.
\end{corollary}

\begin{proof}
First, notice that any chain can traverse any rank $i$ at most once. Hence, the maximum length chain has to traverse each rank once and only once, meaning that it starts at rank $0$ (given by  $\uv=(0,\dots,0)$) and walks through one element in each rank $i$ till it reaches the maximum element (which is $\uv=(1,\dots,1)$). 

Second, in order to compute the length of this chain we use the rank function $\rho$ for the maximum element, $\rho(1,\dots,1)=\sum\limits_{i=1}^mi=\binom{m+1}{2}.$
\end{proof}

\begin{proposition}\label{pr:algo_max_chain}
Define the following algorithm
\begin{algorithm}[!h]
\SetKwInOut{Input}{Input}
\SetKwInOut{Output}{Output}
\Input{ A strictly positive integer $m$}
\Output{ A maximum length chain $\mathcal{S}$ for $\Ps$}
 
 $\mathcal{S}=\{\}$\;
$ k = 0$\;
 \For{$j$ from $1$ to $m$}{
   \For {$i$ from $0$ to $m-j$}{
    $\mathcal{S}=\mathcal{S}\cup \{k+2^i\}$\;}
    $k=k+2^{m-j}$\;
   }
   \textbf{Return} $\mathcal{S}$ in binary\;
% \caption{Building a maximum length chain for $\Ps$.}\label{algo:1}
\end{algorithm}

This algorithm constructs a maximum length chain  $\mathcal{S}$ for $\Ps$. 
\end{proposition}
\begin{proof}
This algorithm builds a set $\mathcal{S}=\bigcup_{j=1}^{m}\mathcal{S}_j,$ where $\mathcal{S}_j=\{k_j+2^i\;|\; 0\leq i\leq m-j\},$ with $k_j=\sum_{l=1}^{j}2^{m-l+1}\pmod{2^m}.$ The first set $\mathcal{S}_1$ is the set of all possible powers of $2$. The second set is the set of all integers that can be written as sum of $2^{m-1}$ plus any other power of $2$, and so on. Notice that the sets $\mathcal{S}_j$ are pairwise disjoints.
By simple inspection of the sets $\mathcal{S}_j$ we have 
\begin{equation}
 \forall \; 1\leq j\leq m \quad \#\mathcal{S}_j=m+1-j.
\end{equation}

Hence, we obtain
\begin{equation*}
\#\mathcal{S}=\sum_{j=1}^{m}\#\mathcal{S}_j=\sum_{j=1}^{m}(m+1-j)=\binom{m+1}{2}.
\end{equation*}

We still need to check whether $\mathcal{S}$ is a chain in $\Ps$. If we expand each element in $\mathcal{S}_j$ for any fixed $j$, we notice that they are totally ordered with respect to $\preceq_H$. On top of that, the first element of $\mathcal{S}_{j+1}$ is larger than the last element of $\mathcal{S}_j$ with respect to the order $\preceq_S$, which concludes the proof.
\end{proof}

%In Fig. \ref{fig:chains} we present the chains constructed using the algorithm from Proposition \ref{pr:algo_max_chain} for $m=5$ and $m=6.$%, and $m=12$.

   \medskip
   
\subsubsection{The middle of the poset}
\begin{corollary}Let $m$ be a strictly positive integer. Then the middle of $\Ps$ is either $\binom{m+1}{2}/2$ (when $m$ is a multiple of $4$), or  $\lfloor\binom{m+1}{2}/2\rfloor$ and $\lceil\binom{m+1}{2}/2\rceil$ (otherwise).
\end{corollary}
The middle of the poset corresponds to the sequence of maximum numbers of subsets of $\{1,2,\dots,m\}$ that share the same sum  (\href{https://oeis.org/A025591}{A025591} \text{ in } \cite{OEIS}) : 
\[1, 1, 1, 2, 2, 3, 5, 8, 14, 23, 40, \dots\]

%PRoposition middle rank elements 

\begin{proposition}\label{pr:middle_u}
Let $k$ and $m$ be strictly positive integers with $\uv\in \{0,1\}^m$ such that $$\uv=\left\{
\begin{array}{lllll}
(0^{k}&1^{2k}&0^{k})&& m=4k\\
(0^{k+1}&1^{2k}&0^{k})&& m=4k+1\\
(0^{k}&1^{2k}&0^{k+1}&1)& m=4k+2\\
(0^{k}&1^{2k+2}&0^{k+1})&& m=4k+3\\
\end{array}
\right..
$$

Then, both $\uv$ and $\overline{\uv}$ are middle rank elements of $\Ps$. 
\end{proposition}

Notice that Proposition \ref{pr:middle_u} can be restated in an equivalent form by using $\supp{\uv}$. For example, when $m=4k$ one would rather say that a middle rank element of $\Ps$ is $\uv$ with $\supp{\uv}=\{k+1,\dots,2k\}.$ One can check that $\uv$ is a middle rank element simply by computing the sum of the elements in $\supp{\uv}.$

\begin{proof} 
We prove each possible case separately.
\begin{itemize}
\item $m=4k$ with $k\ge 1.$ In this case the middle rank equals $k(4k+1)=4k^2+k.$ The rank of $\uv=0^{k}1^{2k}0^{k}$ is \[\sum\limits_{i=k+1}^{3k}i=2k^2+2k(2k+1)/2=4k^2+k.\]   
\item $m=4k+1$ with $k\ge 1.$ The middle rank equals $(4k+1)(4k+2)/4=4k^2+3k+1/2.$ So we have two cases, either $4k^2+3k$ or $4k^2+3k+1.$ The rank of $\uv=0^{k+1}1^{2k}0^{k}$ is \[\sum\limits_{i=k+2}^{3k+1}i=4k^2+3k.\] 
\item $m=4k+2$ with $k\ge 1.$ The middle rank equals $(4k+2)(4k+3)/4=4k^2+5k+3/2.$ So we have two cases, either $4k^2+5k+1$ or $4k^2+5k+2.$ The rank of $\uv=0^{k}1^{2k}0^{k+1}1$ is \[4k+1+\sum\limits_{i=k+1}^{3k}i=4k^2+5k+1.\] 
\item$m=4k+3$ with $k\ge 1.$ The middle rank equals $(4k+3)(4k+4)/4=4k^2+7k+3.$  The rank of $\uv=0^{k}1^{2k+2}0^{k+1}$ is \[\sum\limits_{i=k+1}^{3k+2}i=4k^2+7k+3.\] 
\end{itemize}
The fact that $\overline{\uv}$ is also an element of the middle of the poset follows from $\supp{\uv}+\supp{\overline{\uv}}=\binom{m+1}{2}.$
\end{proof}

\medskip

\subsubsection{Maximum length antichain}

Because $\Ps$ is \emph{Sperner}, the maximum length of an antichain is given by the maximum of $\#\Po_i,$ which is achieved by the middle rank. 

Using a well-known theorem by Dilworth \cite{D_1950}, we are able to determine the minimum number of chains in which $\Ps$ can be partitioned.
\begin{theorem}[\cite{D_1950}]\label{thm:dilworth}
The minimum number of chains in which the elements of a poset $\Po$ can be partitioned is equal to the maximum number of elements of an antichain of $\Po$.
\end{theorem}

\begin{example}\label{ex:max_antichain} Maximum antichains.
~
\begin{itemize}
\item $m=5$  $$\mathcal{P}_7=\{
[4, 2, 1],
[4, 3],
[5, 2]\}\quad\mathrm{ and }\quad\mathcal{P}_8=\{
[4, 3, 1],
[5, 2, 1],
[5,3]\}.$$ 

\item $m=6$  $$\mathcal{P}_{10}=\{
[4, 3, 2, 1],
[5, 3, 2],
[5, 4, 1],
[6, 3, 1],
[6, 4]\}\quad\mathrm{ and }\quad\mathcal{P}_{11}=\{
[5, 3, 2, 1],
[5, 4, 2],
[6, 3, 2],
[6, 4, 1],
[6, 5]\}.$$

\item $m=7$  $$\mathcal{P}_{14}=\{
[5, 4, 3, 2],
[6, 4, 3, 1],
[6, 5, 2, 1],
[6, 5, 3],
[7, 4, 2, 1],
[7, 4, 3],
[7, 5, 2],
[7, 6, 1]\}.$$

\item $m=8$  \begin{align*}
\mathcal{P}_{18}=\{&
[6, 5, 4, 2, 1],
[6, 5, 4, 3],
[7, 5, 3, 2, 1],
[7, 5, 4, 2],
[7, 6, 3, 2],
[7, 6, 4, 1],
[7, 6, 5],\\
&[8, 4, 3, 2, 1],
[8, 5, 3, 2],
[8, 5, 4, 1],
[8, 6, 3, 1],
[8, 6, 4],
[8, 7, 2, 1],
[8, 7, 3]\}.\\
\end{align*}
\end{itemize}
\end{example}

A consequence of Theorem \ref{thm:dilworth} is that the minimum number of chains in which $\Ps$ can be partitioned equals $\#\Po_{\binom{m+1}{2}/2}.$ Now, since $\Ps$ is isomorphic to $M(n)$ (see \cite{S_91}), in order to determine the elements of $\mathcal{P}_{\binom{m+1}{2}/2}$ we can use an algorithm that generates subsets of $\{1,\dots, m\}$, such that the sum of their elements equals $\binom{m+1}{2}/2$. Example \ref{ex:max_antichain} illustrates this for several values of $m.$ Notice that, the compositions $\uv$ given by Proposition \ref{pr:middle_u} are also elements of these antichains. %in \ref{ex:max_antichain}.  

\section{Most Reliable MMNs}\label{sec:applic_poset}
\subsection{Uniformly most reliable networks}

\setcounter{theorem}{36}
\begin{definition}
We say that $\Ns^*\in \NS{n}$ is {\it UMR-MMN} if for any $\Ns\in  \NS{n}$ we have 
\begin{equation}
\Rel(\Ns^*)\ge \Rel(\Ns)\quad \forall \; p\in [0,1].\label{eq:Bosch}
\end{equation} 
\end{definition}

This definition is not identical to that of Boesch et al. \cite{1991_BLS}, as although eq. \eqref{eq:Bosch} is the same, the set of networks is different. For Boesch et al. the domain is represented by the set of simple graphs with $n$ edges and $w$ vertices, while in our case the domain is $\NS{n}.$
  
\begin{theorem}
Let $m$ be a positive integer. Then $\Css^{\uv}$, with $\uv=\left(1,\dots,1\right)\in \{0,1\}^m$, is UMR-MMN in $\NS{2^m}$. 
\end{theorem}

\begin{proof}
This follows directly from the fact that $\uv=\left(1,\dots,1\right)\in \{0,1\}^m$ is the supremum of $\Ps$.
\end{proof}
\subsection{Heaviside most reliable networks}

\begin{definition}
We say that $\Ns^*\in \NS{n}$ is {\it HMR-MMN} with $p_0\in (0,1)$ if for any $\Ns\in  \NS{n}$ we have 
\begin{equation}\label{eq:Heavi_MR_1}
\Rel(\Ns^*)\le \Rel(\Ns)\quad \forall \; p\in [0,p_0) 
\end{equation}
and
\begin{equation}\label{eq:Heavi_MR_2}
\Rel(\Ns^*)\ge \Rel(\Ns)\quad \forall \; p\in [p_0,1]. 
\end{equation} 
\end{definition}
This means that an MMN is an HMR-MMN if it is very close to both the minimum of the poset, for a particular range of values, as well as to the maximum of the poset, for the remaining range of values.%, is an HMR-MMN. %The following lemma recalls the minimum and the maximum of the poset.  
\begin{lemma}\label{lem:max_min_MMS}
Let $m$ be a strictly positive integer. For any $\Ns\in \NS{2^m}$ we have
\begin{equation*}\Rel(\Css^{(0,\dots,0)})<\Rel(\Ns)\quad\forall \; p\in (0,1]; 
\end{equation*}
\begin{equation*}\Rel(\Ns)<\Rel(\Css^{(1,\dots,1)})\quad\forall \; p\in [0,1). 
\end{equation*}

\end{lemma}
\begin{theorem}
Let $m$ be a strictly positive integer. Then there is no HMR-MMN for $p_0\in (0,1)$.
\end{theorem}

\begin{proof}
By Lemma \ref{lem:max_min_MMS} we have that the only composition that satisfies eq. \eqref{eq:Heavi_MR_1} is $\Css^{(0,\dots,0)}$. We also have that the only composition that satisfies eq. \eqref{eq:Heavi_MR_2} is $\Css^{(1,\dots,1)}$. Unless $\Css^{(0,\dots,0)}$ equals $\Css^{(1,\dots,1)}$, it is impossible to have an MMN that satisfies both eq. \eqref{eq:Heavi_MR_1} and eq. \eqref{eq:Heavi_MR_2}. So unless $m=0$  there is no HMR-MMN. 
\end{proof}

\subsection{Optimality of MMNs}
\subsubsection{Motivations}
Since HMR-MMNs do not exist, we re-define optimality as follows: establish how close $\Rel(\Ns;p)$ is to $\theta(p-p_0)$. We restrict our search only to square MMNs ($w=l=\sqrt{n}=2^{m/2}$), and support this choice by several arguments. 

One argument is given in \cite{2018_DCHGB_J,2018_DCHGB(1)}, where the authors have proposed several FOMs such as: the steepness of the reliability polynomials, and their variation in a symmetric interval with respect to $p_0=0.5$. Those simulations, as well as our own simulations, have shown that square MMNs come ``closer'' to $\theta(p-0.5)$ that non-square MMNs. Still, these have been verified only for small values of $l$ and $w.$% and there is yet no general proof. 

Another argument is a combinatorial one. Suppose that one would randomly choose from the set of all MMNs of size $n$.  The question one should ask is: Do square MMNs appear with higher probability? 
\begin{proposition}Let $m$ be a strictly positive integer. Then we have 
\begin{equation}
\lim_{m\to\infty}\dfrac{\#\NN{2^{m/2}}{2^{m/2}}}{\#\NS{2^m}}=1.
\end{equation}
\end{proposition} 
\begin{proof} 
Using a known result about the cardinality of the two sets we have
\begin{align*}
\dfrac{\#\NN{2^{m/2}}{2^{m/2}}}{\#\NS{2^m}}&=\dfrac{2^{(2^{m/2}-1)(2^{m/2}-1)}}{\sum\limits_{i=0}^{m}2^{(2^i-1)(2^{m-i}-1)}} \ge \dfrac{2^{(2^{m/2}-1)^2}}{2^{(2^{m/2}-1)^2}+m\cdot2^{(2^{m/2-1}-1)(2^{m/2+1}-1)}}\\
\dfrac{\#\NN{2^{m/2}}{2^{m/2}}}{\#\NS{2^m}}&\ge\dfrac{1}{1+m\cdot 2^{2^{m/2}-2^{m/2-1}-2^{m/2+1}+2^{m/2}}}=\dfrac{1}{1+\frac{m}{2^{m/2-1}}}.\\
\end{align*}

We used here the fact that the sum can be upper bounded by the middle term plus $m$ times the previous term. In order to verify this, one has to check whether \begin{equation}\forall \; 0\le i\le \frac{m}{2}-1, \; 2^{(2^i-1)(2^{m-i}-1)}\le2^{(2^{m/2-1}-1)(2^{m/2+1}-1)}.\end{equation}
Tacking logarithms of both sides and expanding we obtain
\begin{equation}\forall 0\le i\le \frac{m}{2}-1\;,2^{m-i}+2^i-2^{m/2-1}-2^{m/2+1}\ge 0.
\end{equation}
The left part of the inequality can be viewed as an increasing function of $i$. Since for $i=0$ this is positive, as long as $m\ge 2$, the proof is concluded.  
\end{proof}

\subsubsection{Theoretical results for square MMNs} A first result that we prove involves square PoS, square hammocks and square SoP.

\begin{corollary}[\cite{DCB18}]
Let $m$ be an even positive integer. We have 
\begin{equation*}\Rel\left(\Css^{\left(1^{m/2}0^{m/2}\right)}\right)\le \Rel\left(\Hs{2^{m/2}}{2^{m/2}}\right)\le  \Rel\left(\Css^{\left(0^{m/2}1^{m/2}\right)}\right)
\end{equation*}
while there are $m^2/2$ ranks in $\Ps$ between the square PoS and the square SoP.
\end{corollary}

This follows directly from Proposition \ref{pr:sop_ham_pos} by taking $i=m/2$.

\begin{corollary}
Let $m$ be an even positive integer, then \begin{equation}\Rel\left(\Css^{\left(1^{m/2-1}0^{m/2-1}10\right)}\right)\le \Rel\left(\Hs{2^{m/2}}{2^{m/2}}\right).\end{equation}
\end{corollary}

This is a direct consequence of Proposition \ref{pr:sop_comp_ham} for $i=m/2-1.$ The result represents an improvement, as it reduces the number of ranks to be analyzed from $m^2/2$ to $m^2/2-m/2-1$. This is still a large number of elements. On one hand, since $\Ps$ is unimodal and symmetric it follows that the maximum cardinality rank is given by the middle of the poset. On the other hand, for $\uv=(1^{m/2}0^{m/2}10)$ we have $\rho(\uv)<\binom{m+1}{2}/2<\rho(\overline{\uv})$. 

All of these arguments are supporting our choice to analyze only square compositions in the middle of the poset.

\begin{theorem}[\cite{2013_S}]\label{thm:midd_pos_asymp}
Let $m$ be a strictly positive integer. 
The middle of $\Ps$, has cardinality \begin{equation}
\#\Po_{\binom{m+1}{2}/2}=\sqrt{\dfrac{6}{\pi}}\dfrac{2^m}{m^{3/2}}\left(1+o(1)\right), \text{ when }m\to \infty.\label{eq:asympt_antichain}
\end{equation}
\end{theorem}

A direct consequence of Theorem \ref{thm:midd_pos_asymp} is that searching for HMR compositions requires computing the reliability polynomials for only $n/\log^{3/2}(n)$ compositions. Here, since we restricted the study to square MMNs, we need to determine how many compositions in the middle of $\Ps$ are square. 

Our simulations are supporting the intuition that all the compositions in the middle of the poset are either square or close to square. This was verified for $6\leq m\leq 13$. The simulations also showed that roughly half of the compositions in the middle of the poset are square. We conjecture this to be true in general.
\begin{corollary}\label{cor:47}
Let $m$ be a strictly positive integer. When $m\to \infty$ we have
\begin{equation}\dfrac{\#\left(\mathcal{P}_{\binom{m+1}{2}/2}\cap\CN{2^{m/2}}{2^{m/2}}\right)}{\#\CN{2^{m/2}}{2^{m/2}}}=\sqrt{\dfrac{3}{4}}\dfrac{1}{m}\left(1+O\left(\dfrac{1}{m}\right)\right).\end{equation}
\end{corollary}

\begin{proof}
 To estimate the numerator we use Theorem \ref{thm:midd_pos_asymp} with the assumption that half of the elements in the middle of the poset are square MMNs. For the denominator, when $m\to \infty$ we have
\[\binom{m}{m/2}=\sqrt{\dfrac{2}{m\pi}}2^m\left(1-O\left(\dfrac{1}{m}\right)\right), \]
\end{proof}

A straightforward interpretation of Corollary \ref{cor:47} is that the ratio of the number of square compositions in the middle of the poset over the number of all square compositions decreases as $\log(n).$% where $n$ is the number of devices. 
\subsubsection{Simulation results}
\medskip

For several even values of $m$ we have computed the cardinality of the following sets: 
\begin{itemize}
\item the set of compositions;
 \item the set of square compositions;
 \item the middle of the $\Ps$;
 \item the set of square compositions in the middle of $\Ps$.
\end{itemize}
 These results can be seen in Table \ref{tab:fin}~.
Our successive optimizations have reduced the cardinality of the sets to be analyzed, in particular, when $m=12$ there are $58$ square compositions in the middle of the poset out of:
\begin{itemize}
\item $4096$ possible compositions ($\sim 1.4\%$);
\item 924 square compositions ($\sim 6.2\%$);
\item 124 compositions in the middle of the poset ($\sim 46.7\%$).
\end{itemize}

We can use duality in order to decrease even further the size of the sets \textit{by a factor of} 2. More exactly, if $\uv\in \mathcal{P}_i$ then we know that $\overline{\uv}\in \mathcal{P}_{n-i}.$ 
With all of these optimizations at hand we have computed all the square compositions in the middle of $\Ps$ for $m=6$. We selected half of them by duality, and have recovered the same results as in \cite{2018_DCHGB_J}. These were obtained by computing the reliability polynomials of only 3 MMNs instead of $64$ as in \cite{2018_DCHGB_J}.

\begin{table}[!h]
\centering
\caption{Cardinality of the set of compositions, square compositions, the middle of $\Ps$, \\square compositions in the middle of $\Ps$, and the ratio of square compositions in the middle of  $\Ps$ over all square compositions.}\label{tab:fin}%, for $ m\in \{4,6,8,10,12\}$.}\label{fig:exp_res_comp_square}
\resizebox{\textwidth}{!}{
\begin{tabular}{|c||c|c||c|c||c|}
\hline
%&&&&&\\
$m$&$\#\CS{2^m}$&$\#\CN{2^{m/2}}{2^{m/2}}$&$\#\mathcal{P}_{\binom{m+1}{2}/2}$&$\#\left(\mathcal{P}_{\binom{m+1}{2}/2}\cap\CN{2^{m/2}}{2^{m/2}}\right)$&$\dfrac{\#\left(\mathcal{P}_{\binom{m+1}{2}/2}\cap\CN{2^{m/2}}{2^{m/2}}\right)}{\#\CN{2^{m/2}}{2^{m/2}}}$\\
%&&&&&\\
\hline\hline
4&16&6&2&2&1\\
\hline
6&64&20&$\{5,5\}$&$\{3,3\}$&0.3\\
\hline
8&256&70&14&8&0.11\\
\hline
10&1024&252&$\{40,40\}$&$\{20,20\}$&0.16\\
\hline
12&4096&924&124&58&0.06\\
\hline
\end{tabular}
}
\end{table}

\section{Conclusions and Perspectives}

In this article, we have described the structure of a poset on the set of compositions of series and parallel two-terminal networks. We have used this structure to derive results on the existence of UMR-MMNs and HMR-MMNs. 
 
There are several directions for extending and improving on the results reported here. The first one is related to the poset of reliability for the set of all MMNs. We have set up here the starting point by defining $\preceq_M$. By means of this large poset we are working on a formal proof that hammocks are the closest MMNs to $\theta(p-0.5).$ 
The second one pertains to other forms of symmetries that could potentially reduce the computations for finding those network closest to $\theta(p-0.5).$ It is to be mentioned that different networks might lead to identical reliability polynomials, e.g., if we swap the  two terminals we obtain different networks having identical polynomials. Also a finer ordering of the reliability polynomials of the compositions would enable a more efficient algorithm for finding the optimal networks. 
%
%\appendices
\begin{appendix}
\section{Appendix}\label{app:A}

\begin{proposition}\label{pr:order_1}
Let $\uv$ and $\vv$ be two binary vectors of size $m.$ Then we have
$\uv\preceq_S\vv\Rightarrow \uv\le \vv.$%\Rel(\Css^{\uv})\le \Rel(\Css^{\vv}). \]
\end{proposition}

In order to prove this proposition we need the following lemma.
\begin{lemma} \label{lem:compo_bhattacharya}Let $s$ be a strictly positive integer and for $i\in\{1,\dots,s\}$,
  let $l_i, l_i^*$ be increasing functions from $[0,1]\to [0,1]$ such
  that $\forall p\in [0,1], \quad l_i^*(p)\le l_i(p)$.  Let
  $f=l_1\circ\dots\circ l_s$ and $f^*=l_1^*\circ\dots\circ
  l_s^*$, then $\forall p\in [0,1], \quad f^*(p)\le f(p).$
\end{lemma}

This lemma can be easily proved by induction. With this result at hand we can prove Proposition \ref{pr:order_1}.

\begin{proof}[Proof of Proposition \ref{pr:order_1}]
First notice that $\Rel(\Css^{(0)})\leq \Rel(\Css^{(1)}).$ Then let $\uv$ and $\vv$ such that $\supp{\uv}\subset \supp{\vv}.$  By Lemma \ref{lem:compo_bhattacharya} the result holds.
\end{proof}

\begin{proposition}\label{pr:order_2}
Let $\uv$ and $\vv$ be two binary vectors of size $m$ with $|\uv|=|\vv|$. Then we have
$\uv\preceq_H\vv\Rightarrow\uv\le \vv.$% \Rel(\Css^{\uv})\le \Rel(\Css^{\vv}). \]
\end{proposition}

 We will first prove a slightly weaker claim which provides a building block for the final proof.

\begin{lemma}\label{lem:pos_bit}
  Let $1\leq s<m$ and $\uv\in \{0,1\}^m$ be such that $\supp{\uv}=\{{j_1},\dots ,j_{s}\}$ with $1\le j_1<\dots<j_s\le
  m.$ Now let $\uv^*$ be such that $\supp{\uv^*}=\{{j_1},\dots,j_i,j^{*}_{i+1},j_{i+2},\dots,j_{s}\}$ with $j_i\le j^*_{i+1}\le j_{i+1}$. Then $\uv^*\preceq_H \uv$ and 
  $\uv^*\leq \uv.$
 \end{lemma}
 
 \begin{proof}
 From Theorem \ref{thm:reliab_comp} we have 
 
 \begin{align*}\Rel(\Css^{\uv})&=f_1\circ \Rel(\Css^{(0)})^{j_{i+1}-j_{i+1}^*}\circ \Rel(\Css^{(1)})\circ f_2\\
  \Rel(\Css^{\uv^*})&=f_1\circ \Rel(\Css^{(1)})\circ \Rel(\Css^{(0)})^{j_{i+1}-j_{i+1}^*}\circ f_2\\
  \end{align*}
  where $f_1=\Rel(\Css^{(0)})^{j_1}\circ \Rel(\Css^{(1)})\circ \cdots\circ \Rel(\Css^{(0)})^{j_{i}-j_{i-1}-1}
\circ \Rel(\Css^{(1)})\circ \Rel(\Css^{(0)})^{j_{i+1}^*-j_{i}-1}$
%\end{align*} 
and $f_2=\Rel(\Css^{(0)})^{j_{i+2}-j_{i+1}-1}\circ \Rel(\Css^{(1)}) \circ \dots 
\circ \Rel(\Css^{(0)})^{m-j_s-1}.$
%\end{align*}
  
Notice that $\Rel(\Css^{(0)})^{j_{i+1}-j_{i+1}^*}\circ \Rel(\Css^{(1)})$ and $\Rel(\Css^{(1)})\circ \Rel(\Css^{(0)})^{j_{i+1}-j_{i+1}^*}$ are the reliability polynomials of a SoP, respectively a PoS of $w=2$ and $l=2^{j_{i+1}-j_{i+1}^*}.$ Hence by Theorem \ref{thm:order_0} we have $\Rel(\Css^{(1)})\circ \Rel(\Css^{(0)})^{j_{i+1}-j_{i+1}^*} \le \Rel(\Css^{(0)})^{j_{i+1}-j_{i+1}^*}\circ \Rel(\Css^{(1)}).$ Using Lemma \ref{lem:compo_bhattacharya} applied to $\Rel(\Css^{\uv})$ and $\Rel(\Css^{\uv^*})$ we obtain the desired result.   
 \end{proof}

 \begin{proof}[Proof of Proposition \ref{pr:order_2}]
   Let $\uv\preceq_H\vv$ with $\supp{\uv}=\{j_1,\dots,j_s\}$ and $\supp{\vv}=\{k_1,\dots,k_s\}.$ Define for $i=0,\dots,s$ the binary vectors $\uv^{(*i)}$ such that $\supp{\uv^{(*i)}}=\{j_1\dots j_i,
   k_{i+1}\dots k_s\}$. We have $\uv^{(*0)}=\vv$, $\uv^{(*s)}=\uv$, and
   $\uv^{(*(i+1))}\preceq_H \uv^{(*i)}$ verify the hypotheses of the
   previous lemma. Applying the previous lemma $s$ times, we get
   $\Rel(\Css^{\uv})\le \Rel(\Css^{\vv}).$
 \end{proof}
\end{appendix}

\section*{acknowledgements}
Research supported in part by the EU through the European Research Development Fund under the Competitiveness Operational Program ({\it BioCell-NanoART = Novel Bio-inspired Cellular Nano-architectures}, POC-A1-A1.1.4-E-2015 nr. 30/01.09.2016).%

\bibliographystyle{IEEEtran}
\IEEEtriggeratref{35}

\end{document}